\documentclass[
3p,
times
]{elsarticle}

\usepackage[T1]{fontenc}
\usepackage{pdfpages}

\usepackage{amsmath}
\usepackage{amssymb}
\usepackage{amsthm}
\usepackage{csquotes} 
\usepackage{stackengine}
\usepackage{mathtools}
\usepackage{xcolor}
\usepackage[OMLmathsfit]{isomath} 
\usepackage{booktabs}             
\usepackage{tabularx}             
\newcolumntype{R}{>{\raggedleft\arraybackslash}X}  \newcolumntype{L}{>{\raggedright\arraybackslash}X}  \newcolumntype{C}{>{\centering\arraybackslash}X}
\usepackage{subcaption}           

\usepackage{cleveref} 
  \crefformat{algorithm}{Alg.~#2#1#3}        \crefname{algorithm}{Alg.}{Algs.}
  \Crefformat{algorithm}{Algorithm~#2#1#3}   \Crefname{algorithm}{Algorithm}{Algorithms}
  \crefformat{definition}{Def.~#2#1#3}
  \Crefformat{definition}{Definition~#2#1#3}
  \crefformat{equation}{Eq.~(#2#1#3)}
  \Crefformat{equation}{Equation~(#2#1#3)}
  \crefformat{example}{Expl.~#2#1#3}
  \Crefformat{example}{Example~#2#1#3}
  \crefformat{figure}{Fig.~#2#1#3}           \crefname{figure}{Fig.}{Figs.}
  \Crefformat{figure}{Figure~#2#1#3}         \Crefname{figure}{Figure}{Figures}
  \crefformat{lemma}{Lem.~#2#1#3}
  \Crefformat{lemma}{Lemma~#2#1#3}
  \crefformat{page}{p.~#2#1#3}
  \Crefformat{page}{Page~#2#1#3}
  \crefformat{problem}{Prob.~#2#1#3}
  \Crefformat{problem}{Problem~#2#1#3}
  \crefformat{proposition}{Prop.~#2#1#3}
  \Crefformat{proposition}{Proposition~#2#1#3}
  \crefformat{remark}{Rem.~#2#1#3}
  \Crefformat{remark}{Remark~#2#1#3}
  \crefformat{section}{Sec.~#2#1#3}         \crefname{section}{Sec.}{Secs.}
  \Crefformat{section}{Section~#2#1#3}      \Crefname{section}{Section}{Sections}
  \crefformat{table}{Tab.~#2#1#3}           \crefname{table}{Tab.}{Tables}
  \Crefformat{table}{Table~#2#1#3}          \Crefname{table}{Table}{Tables}
  \crefformat{theorem}{Thm.~#2#1#3}         \crefname{theorem}{Thm.}{Thms.}
  \Crefformat{theorem}{Theorem~#2#1#3}      \Crefname{theorem}{Theorem}{Theorems}

\newcommand*{\setE}{\ensuremath{\mathcal{E}}}
\usepackage{upgreek}                                             
\newcommand*{\Gammah}{\Gamma_h}
\usepackage{dsfont}                                              
  
  \newcommand{\IP}{\ensuremath\mathds{P}}
  \newcommand{\IR}{\ensuremath\mathds{R}}

  \renewcommand*{\vec}[1]{{\boldsymbol{#1}}}                     
\DeclareMathAlphabet{\mathbfsf}{\encodingdefault}{\sfdefault}{bx}{n}
  \newcommand*{\vecc}[1]{\mathbfsf{#1}}                          
\newcommand*{\laplace}{\upDelta}                                 
\newcommand*{\grad}{\vec{\nabla}}                                
\newcommand*{\curl}{\vec{\nabla}\times}                          
\renewcommand*{\div}{\vec{\nabla}\cdot}                          
\newcommand*{\strain}{{\boldsymbol{\varepsilon}}}                
\newcommand*{\llbrace}{\lbrace\hspace*{-0.18em}\vert}
\newcommand*{\rrbrace}{\vert\hspace*{-0.18em}\rbrace}
\newcommand*{\avg}[1]{\llbrace{#1}\rrbrace}                      
\newcommand*{\jump}[1]{\left\llbracket{#1}\right\rrbracket}      
\newcommand*{\dd}{\mathrm{d}}                                    
\newcommand*{\abs}[1]{\ensuremath{|#1|}}                         
\newcommand*{\Rey}{\mathrm{Re}}                                  
\newcommand*{\Nel}{N_\mathrm{el}}                                
\newcommand*{\norm}[2]{\|#1\|_{#2}}                              
\newcommand*{\normal}{\vec{n}}                                   
\newcommand*{\transpose}[1]{{#1}^\mathrm{T}}                     
\newcommand*{\on}[2]{\left.#1\right\vert_{#2}}                   
\newcommand{\upwind}[1]{#1^{\uparrow}}                           

\newtheorem{thm}{Theorem}[section]

\newtheorem{proposition}[thm]{Proposition}
\newtheorem{remark}{Remark}[section]

\begin{document}
\begin{frontmatter}

\title{An interior penalty discontinuous Galerkin approach for 3D incompressible Navier--Stokes equation for permeability estimation of porous media}
\author[label1]{Chen Liu}
\author[label2]{Florian Frank}
\author[label1,label3]{Faruk O.~Alpak}
\author[label1]{B{\'e}atrice Rivi{\`e}re\corref{cor1}}\ead{riviere@rice.edu}
\address[label1]{Rice University, Department of Computational and Applied Mathematics, 6100 Main Street, Houston, TX 77005, USA}
\address[label2]{Friedrich-Alexander-Universit\"at Erlangen-N\"urnberg, Department Mathematik, Cauerstra{\ss}e~11, 91058~Erlangen, Germany}
\address[label3]{Shell Technology Center, 3333 Highway 6 South, Houston, TX 77082, USA}

\cortext[cor1]{Corresponding author: B{\'e}atrice Rivi{\`e}re}

\begin{abstract}
Permeability estimation of porous media from direct solving Navier--Stokes equation has a wide spectrum of applications in petroleum industry. In this paper, we utilize a pressure-correction projection algorithm in conjunction with the interior penalty discontinuous Galerkin scheme for space discretization to build an incompressible Navier--Stokes simulator and to use this simulator to calculate permeability of real rock sample. The proposed method is accurate, numerically robust, and exhibits the potential for tackling realistic problems.
\end{abstract}

\end{frontmatter}


\section{Introduction}
In the past decades, digital rock physics (DRP) has undergone rapid development. The modern X-ray microtomography (micro-CT) creates cross-sections of small rock samples on micrometer to millimeter scale by means of X-rays, which subsequently can be used to construct high-resolution model domains by 3D imaging software. This technique provides geometry patterns of rock structure, cf.~\Cref{fig:porousDomain256}. With the increasing availability of computational resources and developments in numerical algorithms for 3D image reconstruction of pore structures, the use of direct numerical simulation for computing effective properties  of a porous medium is playing a key role in the understanding of flows and rock interactions and is complementing time-consuming high-cost laboratory measurements. In addition, compared with traditional empirical formula based estimations, direct numerical simulation methods provide more accurate results, which attract much attention in hydro-geophysics and petroleum engineering fields.
\begin{figure}[ht!]
\centering
\includegraphics[width=0.6\linewidth]{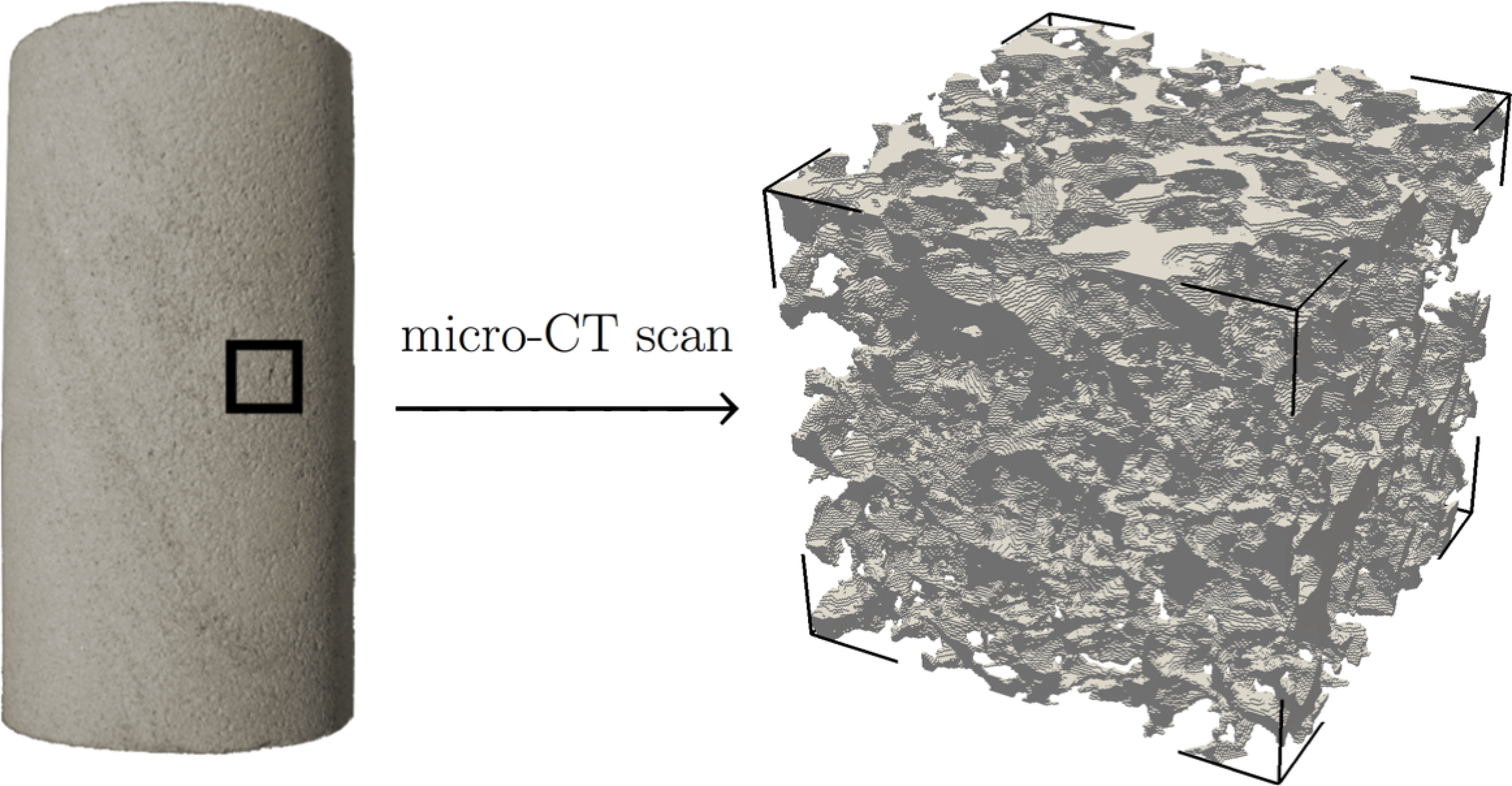}
\caption{A~$256\times256\times256$ voxel\index{voxel} sets describing the fluid space of porous matrices.}
\label{fig:porousDomain256}
\end{figure}
\par 
Estimating permeabilities of rock matrices by direct numerical simulation usually requires large-scale computing to solve incompressible Navier--Stokes equation on complicated computational domains. Therefore an efficient numerical scheme becomes significantly important. The classical approach of solving the saddle-point problem by solving simultaneously the momentum and continuity equations has been well studied \cite{GirRav86,temam2001navier}. However, this approach is not amenable for large systems. After taking both numerical accuracy and computational efficiency into account and balancing all considerations, one preferred technique is to decouple the nonlinear reaction term from the pressure term by using projection algorithms \cite{chorin1968numerical}. The splitting preserves important properties such as incompressibility or consistency with the boundary data. The literature on projection algorithms coupled with continuous finite element methods is vast. We refer the reader to an overview of the splitting techniques in \cite{guermond2006overview,Blasco2004}.
\par
Interior penalty discontinuous Galerkin (IPDG) methods form a popular class of accurate and robust numerical schemes for solving partial differential equations \cite{riviere2008,WarburtonBook}, which are known to be very flexible and have several positive features: local construction of trial and test spaces; curved boundaries and inhomogeneous boundary conditions are naturally handled; a local mass conservation property and suitability for parallelization in the context of large-scale simulations in complex domains. Due to these advantages, IPDG methods are widely used in computational fluid dynamics (CFD), geomechanics, and porous medium simulations. Although, a lot of research has been done for developing stable and high-order accurate solvers for CFD, there are very few works in the literature on the application of IPDG methods to incompressible flows for large scale computing. The use of discontinuous elements is studied in a pressure-correction algorithm in \cite{piatkowski2018stable} and in a velocity-correction projection algorithm in \cite{krank2017high}.  The paper \cite{piatkowski2018stable} employs a modified upwind scheme for the reaction term and a discretization of the projection step that are different from the ones we propose in this paper.
\par
The paper's contribution to the field of pore scale flows is a numerical approach for computing permeabilities of porous media. The porous structure is obtained from micro-CT imaging of the rock sample. IPDG methods combined with a projection algorithm are used to discretized the incompressible Navier--Stokes equations. The nonlinear system is linearized by a Picard splitting and a div--div correction postprocessing technique ensures a solenoidal velocity field. From the computed velocity and pressure fields, averages pressure and flux on inflow and outflow faces are obtained. The proposed approach is validated on benchmark problems. Numerical experiments on real rock images show the robustness of this method for handling 3D large scale simulations on complex geometries.
\par 
Estimating effective properties of rocks, such as absolute permeability, is one important objective of DRP \cite{Dvorkin2011}. There are few papers on the estimation of permeabilities from real rocks. For instance, the work \cite{Andra2013} uses lattice Boltzmann method (LBM) and finite difference method to solve the Stokes equations and to obtain permeabilities of rock samples. More recently, a comprehensive numerical study of various rock samples is provided in \cite{saxena2017references} where permeability values are estimated by various solvers including LBM, continuous finite element, finite volume and fast Fourier transforms.  Comparisons with experimental values show in some cases significant discrepancies. In addition, digital values of permeabilities vary from solver to solver. While the study does not fully address the reasons behind the mismatch, it highlights the need for accurate flow solvers.  To our knowledge, this paper is the first to propose a DG-based solver for computing effective rock properties.
\par
The outline of this paper follows. The incompressible Navier--Stokes model including boundary conditions for inflow, outflow, and fluid--solid interface are described in~\Cref{sec:ModelProblem}. The numerical algorithm is provided in~\Cref{sec:NS:numerical_scheme}. The validation tests and the numerical simulations for permeability estimation are given in~\Cref{sec:NS:numerical_simulation}.

\section{Model Problem}\label{sec:ModelProblem}
Complex fluid phenomena in nature and in industrial applications give rise to a large number of challenging mathematical problems. The Navier--Stokes equations are the classical mathematical equations describing viscous gas and liquid flow. In this section, we introduce the equations for incompressible fluids, discuss the boundary conditions, and nondimensionalize the system equations.

\subsection{Governing equations}
Let $\Omega \subset \IR^3$ denote an open bounded polyhedral domain and let $\normal$ denote the unit outer normal of $\Omega$. Consider an incompressible fluid occupying the spatial domain $\Omega$ over the time interval $(0,T)$, where the unknown variables velocity $\vec{v}$ and pressure $p$ satisfy the Navier--Stokes equations:
\begin{subequations}\label{eq:GoverningEquationsNS}
\begin{align}
\rho_0\,(\partial_t{\vec{v}} + \vec{v}\cdot\grad{\vec{v}}) - \mu_\mathrm{s}\laplace{\vec{v}}
~&=~ -\grad{p} &&\text{in} ~(0,\,T)\times\Omega~, \label{eq:NavierStokesEquation} \\
\div\vec{v} ~&=~ 0 && \text{in} ~(0,\,T)\times\Omega~. \label{eq:IncompressibilityConstraint}
\end{align}
\end{subequations}
Here, we assume a positive constant fluid density $\rho_0$ and a positive constant shear viscosity $\mu_\mathrm{s}$. The diffusion term~$-\laplace{\vec{v}}$ and the convection term~$\vec{v}\cdot\grad{\vec{v}}$ in~\cref{eq:NavierStokesEquation} are simplifications of the more general forms $-2\div{\strain(\vec{v})}$ and $\div{(\vec{v}\otimes\vec{v})}$\,, respectively, where
\begin{equation*}
\strain(\vec{v}) = \frac{1}{2} \big(\grad \vec{v} + \transpose{(\grad\vec{v})}\big)
\end{equation*}
is the deformation tensor. The incompressibility constraint is expressed by \cref{eq:IncompressibilityConstraint}, which is also called the pressure equation. With constraint \cref{eq:IncompressibilityConstraint}, the equivalence of the diffusion and convection terms directly follows from the identities
\begin{align*}
2\div{\strain(\vec{v})} ~&=~ \grad{(\div{\vec{v}})} + \laplace{\vec{v}}~,\\
\div(\vec{v}\otimes\vec{v}) ~&=~ \vec{v}\cdot\grad{\vec{v}} + (\div{\vec{v}})\vec{v}~.
\end{align*}
Note since these operators are mathematically equivalent in the model, one might consider formulating problems more generally. Even though these terms are identical, they lead to different boundary conditions. It turns out that using the more general expressions leads to unexpected behaviors of the velocity field at the outflow boundary in open boundary simulations, cf.~\Cref{sec:Boundary_conditions}.

\subsection{Initial and boundary conditions}\label{sec:Boundary_conditions}
Boundary conditions play an important role in modeling various physical phenomena. In this section, we will consider two different scenarios, i.\,e., a closed boundary model and an open boundary model. For the former case, fluid flow across the boundary of the domain is suppressed. In both cases, \cref{eq:GoverningEquationsNS} is supplemented by the initial condition:
\begin{align}\label{eq:NS:IC}
\vec{v} ~=~ \vec{v}^0 && \text{on} ~\{0\}\times\Omega~.
\end{align}

\paragraph{Closed boundary model}
In this scenario, we assume the domain~$\Omega$ is surrounded by a~solid wall, i.\,e., $\partial\Omega = \partial\Omega^\mathrm{wall}$, where $\partial\Omega^\mathrm{wall}$ denotes the fluid--solid interface.
We employ the no-slip boundary condition (homogeneous Dirichlet condition):
\begin{align}\label{eq:NS:closed_boundary_model_BC}
\vec{v} ~=~ \vec{0} && \text{on}~(0,\,T)\times\partial\Omega~.
\end{align}
Due to the boundary condition \cref{eq:NS:closed_boundary_model_BC}, the pressure $p$ is uniquely defined up to an additive constant. To close this system, we also assume that the mean pressure in $\Omega$ equals zero:
\begin{equation}\label{eq:NS:close_system_mean_p}
\int_\Omega p ~=~ 0~.
\end{equation}
Benefitting from boundary condition \cref{eq:NS:closed_boundary_model_BC}, the closed boundary Navier--Stokes model is an energy dissipative system. For regular enough solutions, a brief proof as follows.
\begin{proposition}\label{thm:NS:energy_dissipation}
The kinetic energy of system of problem $\{\eqref{eq:GoverningEquationsNS},\eqref{eq:NS:IC},\eqref{eq:NS:closed_boundary_model_BC},\eqref{eq:NS:close_system_mean_p}\}$ is non-increasing in time, i.\,e., we have the identity, for all $t\in (0,\, T)$:
\begin{equation}\label{eq:NS:energy_dissipation}
\frac{\dd}{\dd t} \int_\Omega \frac{1}{2}\rho_0\abs{\vec{v}}^2 ~=~ -\int_{\Omega} \mu_\mathrm{s}\grad{\vec{v}}:\grad{\vec{v}} ~\leq~ 0~.
\end{equation}
\end{proposition}
\begin{proof}
Taking the time derivative, based on the assumptions that the order of differentiation and integration can be changed, and by chain rule, we have pointwise in~$t$:
\begin{equation*}
\frac{\dd}{\dd t} \int_\Omega \frac{1}{2}\rho_0\abs{\vec{v}}^2 
~=~ \int_\Omega \frac{\partial}{\partial t} \Big(\frac{1}{2}\rho_0\abs{\vec{v}}^2\Big)
~=~ \int_\Omega \vec{v}\cdot(\rho_0\partial_t\vec{v})~.
\end{equation*}
By \cref{eq:NavierStokesEquation}, it follows
\[
\frac{\dd}{\dd t} \int_\Omega \frac{1}{2}\rho_0\abs{\vec{v}}^2
~=~ \int_\Omega -\rho_0\vec{v}\cdot(\vec{v}\cdot\grad{\vec{v}}) + \int_\Omega \mu_\mathrm{s}\vec{v}\cdot\laplace{\vec{v}} - \int_\Omega \vec{v}\cdot\grad{p}~.
\]
Using the incompressibility constraint \cref{eq:IncompressibilityConstraint} and the boundary condition \cref{eq:NS:closed_boundary_model_BC}, the following identities imply \cref{eq:NS:energy_dissipation}
\begin{align*}
\int_\Omega \vec{v}\cdot(\vec{v}\cdot\grad{\vec{v}})
~&=~ \int_\Omega \Big(\frac{1}{2}\vec{v}\cdot\grad{(\vec{v}\cdot\vec{v})} - \vec{v}\cdot(\vec{v}\times\curl{\vec{v}})\Big) \\
~&=~ \frac{1}{2} \int_\Omega \Big(-\abs{\vec{v}}^2\div{\vec{v}} + \div{(\abs{\vec{v}}^2\vec{v})}\Big)
~=~ \frac{1}{2} \int_{\partial\Omega} \abs{\vec{v}}^2\vec{v}\cdot\normal 
~=~ 0~,\\
\int_\Omega \vec{v}\cdot\laplace{\vec{v}} 
~&=~ \sum_{i=1}^{d} \int_\Omega v^i\laplace{v^i}
~=~ \sum_{i=1}^{d} \int_\Omega \Big(-\grad{v^i}\cdot\grad{v^i} + \div{(v^i\grad{v^i})}\Big) \\
~&=~ -\sum_{i=1}^{d} \int_\Omega \abs{\grad{v^i}}^2 + \sum_{i=1}^{d} \int_{\partial\Omega} v^i\grad{v^i}\cdot\normal
~=~ -\int_\Omega \grad{\vec{v}}:\grad{\vec{v}}~,\\
\int_\Omega \vec{v}\cdot\grad{p} 
~&=~ \int_\Omega \Big(-p\div{\vec{v}} + \div{(p\vec{v})}\Big)
~=~ \int_{\partial\Omega} p\vec{v}\cdot\normal
~=~ 0~.
\end{align*}
\end{proof}
The energy dissipation law still holds in case of changing the diffusion term $-\mu_\mathrm{s}\laplace{\vec{v}}$ in \cref{eq:NavierStokesEquation} to $-2\div\big(\mu_\mathrm{s}\strain(\vec{v})\big)$\,. The modification of the proof can be easily obtained by using the tensor identity 
\begin{equation*}
\div{\big(\mu_\mathrm{s}\strain(\vec{v})\vec{v}\big)} ~=~ \mu_\mathrm{s}\,\strain(\vec{v}):\strain(\vec{v}) + \div{\big(\mu_\mathrm{s}\strain(\vec{v})\big)}\cdot\vec{v}~,
\end{equation*}
Note here, in this general case, $\mu_\mathrm{s}$ may vary in space. 

\paragraph{Open boundary model}
In addition to the impervious fluid--solid interface $\partial\Omega^\mathrm{wall}$, we partition $\partial\Omega$ into the inflow boundary and outflow boundary, which are defined as follows:
\begin{align*}
\partial\Omega^\mathrm{in} ~&=~ \{\vec{x}\in\partial\Omega:~\vec{v}\cdot\normal~<~0\}~,\\
\partial\Omega^\mathrm{out} ~&=~ \partial\Omega \setminus (\partial\Omega^\mathrm{in}\cup\partial\Omega^\mathrm{wall})~.
\end{align*}
The fluid enters the domain $\Omega$ through $\partial\Omega^\mathrm{in}$ and exits through $\partial\Omega^\mathrm{out}$. The unit normal
vector, $\normal$, is outward to $\Omega$. In this model, we prescribe inhomogeneous Dirichlet boundary condition on the inflow boundary, no-slip condition at the fluid--rock interface, and a Neumann-type condition on the outflow boundary
\begin{subequations}\label{eq:NS:open_boundary_model_BC}
\begin{align}
\vec{v} ~&=~ \vec{v}^\mathrm{in} && \text{on}~(0,\,T)\times\partial\Omega^\mathrm{in}~,\label{eq:NS:open_boundary_model_inflow}\\
\vec{v} ~&=~ \vec{0} && \text{on}~(0,\,T)\times\partial\Omega^\mathrm{wall}~,\label{eq:NS:open_boundary_model_wall}\\
\mu_\mathrm{s}\grad{\vec{v}}\,\normal  - p\,\normal ~&=~ \vec{0} && \text{on}~(0,\,T)\times\partial\Omega^\mathrm{out}~.\label{eq:NS:open_boundary_model_outflow}
\end{align}
\end{subequations}
Due to the fact that both \cref{eq:NS:open_boundary_model_inflow} and \cref{eq:NS:open_boundary_model_wall} are Dirichlet-type conditions, we define
\begin{equation*}
\partial\Omega^\mathrm{D} ~=~ \partial\Omega^\mathrm{in}\cup\partial\Omega^\mathrm{wall}~.
\end{equation*}
In addition, we define $\vec{v}_\mathrm{D}$ on $\partial\Omega^\mathrm{in}$ equal $\vec{v}^\mathrm{in}$ and extend $\vec{v}_\mathrm{D}$ on $\partial\Omega^\mathrm{wall}$ by $\vec{0}$.
The condition on $\partial\Omega^\mathrm{out}$ is also called the open boundary condition. The motivation of introducing \cref{eq:NS:open_boundary_model_outflow} is that it enables us to simulate on a bounded computational domain to obtain the results identical to the case of fluid flow over unbounded or a relatively large region. Details for deriving \cref{eq:NS:open_boundary_model_outflow} are provided in \cite{bruneau2000boundary} and this condition has been verified to be well-suited in modeling parallel flows \cite{rannacher1996artificial}.
\par
The open boundary condition implies a zero mean pressure over outflow boundaries, which may cause undesirable effects within the flow region if more than one outlet exists~\cite{rannacher2000finite}.
Another alternative is to use the diffusion term $-2\div\big(\mu_\mathrm{s}\strain(\vec{v})\big)$ instead of $-\mu_\mathrm{s}\laplace{\vec{v}}$ in \cref{eq:NavierStokesEquation}. Then, the boundary condition~\cref{eq:NS:open_boundary_model_outflow} has to be modified to
\begin{align*}
2\mu_\mathrm{s}\strain{(\vec{v})}\,\normal  - p\,\normal ~=~ \vec{0} && \text{on}&~(0,\,T)\times\partial\Omega^\mathrm{out}~.
\end{align*}
Unfortunately, that choice leads to streamlines spreading outward at the outlet. 
\Cref{fig:NS:strain_effect} shows Poiseuille flow in a cylinder for two choices of diffusion operators. In the left figure, the operator
$-\mu_\mathrm{s}\laplace{\vec{v}}$ is employed and in the right figure, the operator $-2\div\big(\mu_\mathrm{s}\strain(\vec{v})\big)$ is used.
We observe unphysical streamlines at the outflow boundary for the right figure.
\begin{figure}[ht!]
\centering
\begin{tabularx}{\linewidth}{@{}C@{}C@{}}
\includegraphics[width=0.4\textwidth]{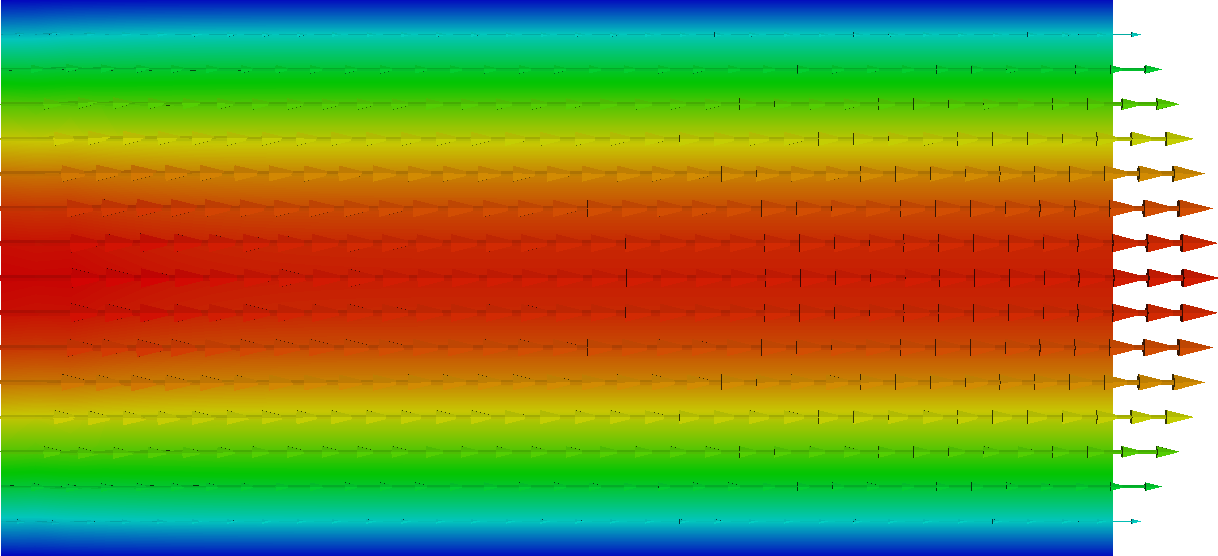} &
\includegraphics[width=0.4\textwidth]{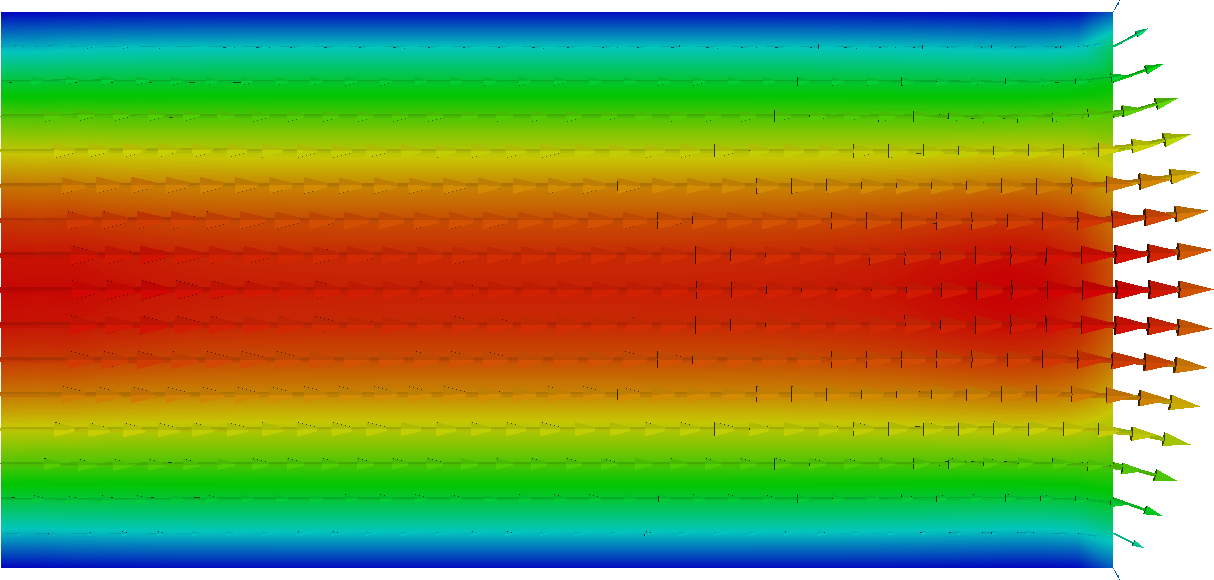}
\end{tabularx}
\caption{Left: open boundary cylindrical simulation with diffusion term $-\mu_\mathrm{s}\laplace{\vec{v}}$. Right: unphysical phenomenon of using deformation tensor $-2\div\big(\mu_\mathrm{s}\strain(\vec{v})\big)$.}
\label{fig:NS:strain_effect}
\end{figure}
\par
Throughout this paper, we focus on problems which are defined on channels with only one outlet and governed by \cref{eq:GoverningEquationsNS}. Thus we avoid these drawbacks.

\subsection{Nondimensionalization}
The nondimensionalization process removes physical dimensions from the fluid mechanics model while producing dimensionless parameters which govern a whole family of dynamically similar flows. An advantage of analyzing problems in nondimensional form is the reduction of parameters. A brief review of the physical quantities and their units in the Navier--Stokes model is provided in \Cref{tab:NS_QuantitiesUnits}. 
\begin{table}[ht]\centering
\begin{tabular}{@{}ccc@{~~|~~}ccc@{}}
\toprule
\textbf{Quantity} & \textbf{Symbol} & \textbf{Unit} & 
\textbf{Quantity} & \textbf{Symbol} & \textbf{Unit} \\
\midrule 
time                & $t$              & $\mathrm{s}$                                    &
length              & $x$              & $\mathrm{m}$                                    \\
velocity            & $\vec{v}$        & $\mathrm{m}\,\mathrm{s}^{-1}$                   & 
pressure            & $p$              & $\mathrm{kg}\,\mathrm{m}^{-1}\,\mathrm{s}^{-2}$ \\
mass density        & $\rho_0$         & $\mathrm{kg}\,\mathrm{m}^{-3}$                  &
shear viscosity     & $\mu_\mathrm{s}$ & $\mathrm{kg\,m^{-1}\,s^{-1}}$                   \\
\bottomrule
\end{tabular}
\caption{Physical quantities used in the Navier--Stokes model and its units.}
\label{tab:NS_QuantitiesUnits}
\end{table} 
\par
Let $x_\mathrm{c}$ denote the characteristic length, $t_\mathrm{c}$ denote the characteristic time, $v_\mathrm{c}$ denote the characteristic velocity, and $p_\mathrm{c}$ denote the characteristic pressure, respectively. For simplicity, we still employ the same symbols for dimensionless quantities. With the choice $t_\mathrm{c}=x_\mathrm{c}/v_\mathrm{c}$\,, the dimensionless version of \cref{eq:GoverningEquationsNS} becomes
\begin{subequations}\label{eq:GoverningEquationsNS_ndim}
\begin{align}
\partial_t \vec{v} + \vec{v}\cdot\grad{\vec{v}} -\frac{1}{\Rey}\laplace{\vec{v}} ~&=~ -\frac{1}{\mathrm{X}}\grad{p} &&\text{in}~(0,\,T)\times\Omega~, \label{eq:NavierStokesEquation_ndim} \\
\div\vec{v} ~&=~ 0 && \text{in}~(0,\,T)\times\Omega~, \label{eq:IncompressibilityConstraint_ndim}
\end{align}
coupled with the initial and boundary conditions of the open boundary model in \Cref{sec:Boundary_conditions}:
\begin{align}
\vec{v} ~&=~ \vec{v}^0 && \text{on}~\{0\}\times\Omega~,\label{eq:NS:open_boundary_model_init_ndim}\\
\vec{v} ~&=~ \vec{v}^\mathrm{in} && \text{on}~(0,\,T)\times\partial\Omega^\mathrm{in}~,\label{eq:NS:open_boundary_model_inflow_ndim}\\
\vec{v} ~&=~ \vec{0} && \text{on}~(0,\,T)\times\partial\Omega^\mathrm{wall}~,\label{eq:NS:open_boundary_model_wall_ndim}\\
\grad{\vec{v}}\,\normal - \frac{\Rey}{\mathrm{X}}p\,\normal ~&=~ \vec{0} && \text{on}~(0,\,T)\times\partial\Omega^\mathrm{out}~.\label{eq:NS:open_boundary_model_outflow_ndim}
\end{align}
\end{subequations}
In \cref{eq:GoverningEquationsNS_ndim}, $\Rey$ denotes the Reynolds number. We have:
\begin{equation*}
\Rey ~=~ \frac{\rho_0\,v_\mathrm{c}\,x_\mathrm{c}}{\mu_\mathrm{s}}  ~=~ \frac{\rho_0\,v_\mathrm{c}^2}{\mu_\mathrm{s}\,v_\mathrm{c}/x_\mathrm{c}}~,
\end{equation*}
which means that the Reynolds number expresses the ratio of inertia to viscous forces. The other dimensionless parameter in \cref{eq:GoverningEquationsNS_ndim} is
\begin{equation*}
\mathrm{X} ~=~ \frac{\rho_0\,v_\mathrm{c}^2}{p_\mathrm{c}}~.
\end{equation*}
%

\section{Numerical Scheme}\label{sec:NS:numerical_scheme}
Direct numerical simulation of fluid flow in porous media needs to solve the time-dependent highly nonlinear Navier--Stokes equations. For this reason, we carefully design a numerical method to solve this system on large complex domains by taking both numerical accuracy and computational efficiency into account. 

\subsection{Preliminaries}\label{sec:NS:scheme_preliminaries}
\paragraph{Domain and triangulation}
Let $\setE_h = \{E_k\}$ be a family of conforming non-degenerate (also called regular) meshes of the domain $\Omega$. The parameter $h$ denotes the maximum element diameter. Let $\Gammah$ donate the set of interior faces. For each interior face $e \in \Gammah$ shared by elements $E_{k^-}$ and $E_{k^+}$, we define a unit normal vector $\normal_e$ points from $E_{k^-}$ into $E_{k^+}$. For the face $e$ on boundary $\partial\Omega$, i.\,e., $e = E_{k^-} \cap \partial\Omega$, the normal $\normal_e$ is taken to be the unit outward vector to $\partial\Omega$.
\par
We denote by $\IP_1(E)$ the set of linear polynomials on element $E$ and we define
\begin{align*}
\IP_1(\setE_h) ~=~ \prod_{E_k\in\setE_h} \IP_1(E_k)~.
\end{align*}
The average and jump operator for any vector quantity $\vec{v}$ belonging to $\IP_1(\setE_h)^3$ are defined by
\begin{align*}
\avg{\vec{v}} ~=~ 
\begin{Bmatrix*}[l]
\frac{1}{2}\on{\vec{v}}{E_{k^-}} + \frac{1}{2}\on{\vec{v}}{E_{k^+}} & \text{if}~e = E_{k^-} \cap E_{k^+}~ \\
\on{\vec{v}}{E_{k^-}} & \text{if}~e = E_{k^-} \cap \partial\Omega~ \end{Bmatrix*}~, 
&&
\jump{\vec{v}} ~=~ 
\begin{Bmatrix*}[l]
\on{\vec{v}}{E_{k^-}} - \on{\vec{v}}{E_{k^+}}~~~~~~ & \text{if}~e = E_{k^-} \cap E_{k^+}~ \\
\on{\vec{v}}{E_{k^-}}~~~~~~ & \text{if}~e = E_{k^-} \cap \partial\Omega~ \end{Bmatrix*}~.
\end{align*}
Note for scalar quantities, the jump and average definitions are similar.

\paragraph{DG forms}
The DG discretization of the convection term $\vec{v}\cdot\grad{\vec{v}}$\,, the diffusion term $-\laplace{\vec{v}}$\,, and the pressure term $-\grad{p}$, with the boundary conditions \cref{eq:GoverningEquationsNS_ndim}, are
\begin{subequations}
\begin{align}
\begin{split}\label{eq:NS:DG_convection_a}
a_{\mathcal{C}}(\vec{u},\vec{z},\vec{\theta}) ~=~&
- \sum_{E\in\setE_h}\int_E (\vec{u}\cdot\grad\vec{\theta})\cdot\vec{z} 
- \frac{1}{2} \sum_{E\in\setE_h}\int_E (\div{\vec{u}})\,\vec{z}\cdot{\vec{\theta}}\\
&+ \sum_{e\in\Gammah\cup\partial\Omega^\mathrm{out}}\int_e \avg{\vec{u}\cdot\normal_e}\,\upwind{\vec{z}}\cdot\jump{\vec{\theta}} 
+ \frac{1}{2} \sum_{e\in\Gammah}\int_e \jump{\vec{u}\cdot\normal_e}\avg{\vec{z}\cdot\vec{\theta}}~,
\end{split}\\
b_{\mathcal{C}}(\vec{\theta}) ~=~& 
- \sum_{e\in\partial\Omega^{\mathrm{in}}} \int_e (\vec{v}_\mathrm{D}\cdot\vec{n})(\vec{v}_\mathrm{D}\cdot\vec{\theta})~, \label{eq:NS:DG_convection_b}\\
\begin{split}\label{eq:NS:DG_strain_a}
a_\strain(\vec{v}, \vec{\theta}) ~=~& 
\sum_{E\in\setE_h} \int_E \grad{\vec{v}}:\grad{\vec{\theta}} 
- \sum_{e\in\Gammah\cup\partial\Omega^\mathrm{D}} \int_e\avg{\grad{\vec{v}}\cdot\normal_e}\cdot\jump{\vec{\theta}}\\
&+ \epsilon\sum_{e\in\Gammah\cup\partial\Omega^\mathrm{D}} \int_e\avg{\grad{\vec{\theta}}\cdot\normal_e}\cdot\jump{\vec{v}}
+ \frac{\sigma}{h}\sum_{e\in\Gammah\cup\partial\Omega^\mathrm{D}} \int_e \jump{\vec{v}}\cdot\jump{\vec{\theta}}~,
\end{split}\\
b_\strain(\vec{\theta}) ~=~& 
\epsilon \sum_{e\in\partial\Omega^\mathrm{D}} \int_e (\grad{\vec{\theta}}\cdot\vec{n}_e)\cdot \vec{v}_\mathrm{D} + \frac{\sigma}{h} \sum_{e\in\partial\Omega^\mathrm{D}}\int_e \vec{v}_\mathrm{D}\cdot\vec{\theta}~, \label{eq:NS:DG_strain_b}\\
b_{\mathcal{P}}(p,\vec{\theta}) ~=~&
\sum_{E\in\setE_h} \int_E p\div{\vec{\theta}}
-\sum_{e\in\Gammah\cup\partial\Omega^\mathrm{D}} \int_e \avg{p}\jump{\vec{\theta}\cdot\normal_e}~. \label{eq:NS:DG_pressure}
\end{align}
For the upwind flux in \cref{eq:NS:DG_convection_a}, we define  the upwind quantity $\upwind{\vec{z}}$ of the vector quantity $\vec{z}$ on a face $e$ with normal $\vec{n}_e=\on{\vec{n}}{E_{k^-}}$ by
\begin{align*}
\on{\vec{z}^\uparrow}{e\in\Gammah} =
\begin{Bmatrix*}[l]
\on{\vec{z}}{E_{k^-}} & \text{if}~\avg{\vec{u}}\cdot\vec{n}_e \geq 0~ \\
\on{\vec{z}}{E_{k^+}} & \text{if}~\avg{\vec{u}}\cdot\vec{n}_e < 0~ \end{Bmatrix*}
&& \text{and} &&
\on{\vec{z}^\uparrow}{e\in\partial\Omega} = 
\begin{Bmatrix*}[l]
\on{\vec{z}}{E_{k^-}} & \text{if}~\vec{u}\cdot\vec{n}_e \geq 0~ \\
\vec{0} & \text{if}~\vec{u}\cdot\vec{n}_e < 0~ \end{Bmatrix*}~.
\end{align*}
In addition, the discretization of the elliptic operator $-\laplace{p}$, which will be employed in the
pressure correction step, is 
\begin{align}
\begin{split}
a_{\mathrm{ellip}}(p,\chi) ~=~&
\sum_{E\in\mathcal{E}_h} \int_E \grad p \cdot \grad \chi
-\sum_{e\in\Gammah\cup\partial\Omega^\mathrm{out}} \int_e \avg{\grad p \cdot \normal_e} \jump{\chi}\\
&+\epsilon\sum_{e\in\Gammah\cup\partial\Omega^\mathrm{out}} \int_e \avg{\grad \chi \cdot \normal_e} \jump{p} + \frac{\sigma}{h} \sum_{e\in\Gammah\cup\partial\Omega^\mathrm{out}}\int_e \jump{p}\jump{\chi}~.
\end{split}\label{eq:NS:DG_diffusion_a}
\end{align}
\end{subequations}
The choice of the parameter $\epsilon$ in \cref{eq:NS:DG_strain_a,eq:NS:DG_strain_b,eq:NS:DG_diffusion_a} yields the NIPG method ($\epsilon=1$) and the SIPG method ($\epsilon=-1$). For NIPG method, the penalty parameter $\sigma$ can be set to $1$ and for SIPG method a large enough value is needed. A discussion about selecting the value of $\sigma$ is provided in \cite{shahbazi2005explicit,EpshteynRiviere2007}. For the derivation and properties of these DG forms, we refer to \cite{riviere2008,GiraultRiviereWheeler2005} and the references herein.

\subsection{Time discretization}\label{sec:NS:time_discretization}
In \cref{eq:GoverningEquationsNS_ndim}, the velocity and pressure are coupled by the incompressibility constraint. The pressure-correction projection algorithm, which splits dynamics and incompressibility, is especially convenient for the large-scale three-dimensional numerical simulations. We now describe the splitting in time of the equations.
\par
Uniformly partition $[0,\,T]$ into $N$ subintervals and let $\tau$ denote the time step length. For any $1 \leq n \leq N$\,, given the solution at the previous time,  $(\vec{v}^{n-1},p^{n-1},\phi^{n-1})$, the algorithm yields the solution at the current time in three steps. In the first velocity step, we employ Picard splitting to treat the nonlinear convection term and compute $\vec{v}^n$ by solving the system:
\begin{subequations}\label{eq:NS:time_dis_velocity_step}
\begin{align}
\vec{v}^n + \tau\vec{v}^{n-1}\cdot\grad{\vec{v}^n} - \frac{\tau}{\Rey}\laplace{\vec{v}^n} ~&=~ \vec{v}^{n-1} - \frac{\tau}{\mathrm{X}}\grad{(p^{n-1}+\phi^{n-1})} && \text{in}~\Omega~,\\
\vec{v}^n ~&=~ \vec{v}^\mathrm{in}(t^n) && \text{on}~\partial\Omega^\mathrm{in}~,\\
\vec{v}^n ~&=~ \vec{0} && \text{on}~\partial\Omega^\mathrm{wall}~,\\
\grad{\vec{v}}^n \normal - \frac{\Rey}{\mathrm{X}}p^{n-1}\, \normal ~&=~ \vec{0} && \text{on}~\partial\Omega^\mathrm{out}~.
\end{align}
\end{subequations} 
Then, in the pressure projection step, we solve a Poisson problem to compute the correction $\phi^n$: 
\begin{subequations}\label{eq:NS:time_dis_pressure_step}
\begin{align}
-\laplace{\phi^n} ~&=~ -\frac{\mathrm{X}}{\tau}\div{\vec{v}^n} && \text{in}~\Omega~,\\
\grad{\phi^n}\cdot\normal ~&=~ 0 && \text{on}~\partial\Omega^\mathrm{D}~,\\
\phi^n ~&=~ 0 && \text{on}~\partial\Omega^\mathrm{out}~.
\end{align}
\end{subequations}
Finally, in the postprocessing step, we update the pressure $p^n$ and obtain a solenoidal velocity field $\vec{u}^n$: 
\begin{subequations}\label{eq:NS:time_dis_post_step}
\begin{align}
p^n ~&=~ p^{n-1} + \phi^n - \frac{\mathrm{X}}{\Rey}\,\div{\vec{v}^n}~,\label{eq:NS:time_dis_post_step_p}\\
\vec{u}^n ~&=~ \vec{v}^n - \frac{\tau}{\mathrm{X}}\,\grad{\phi^n}~.\label{eq:NS:time_dis_post_step_u}
\end{align}
\end{subequations}
Here, for the initial condition, we set $p^0 = 0$ and $\phi^0 = 0$. Note that both $\vec{v}^n$ and $\vec{u}^n$ are approximations of the velocity and $\vec{v}^n$ satisfies the Dirichlet boundary condition while possibly violating an incompressibility constraint, whereas $\vec{u}^n$ is solenoidal but satisfies the Dirichlet condition only in the normal direction.

\subsection{Space discretization}\label{sec:NS:space_discretization}
\paragraph{Standard algorithm}
We use IPDG method for spatial discretization for each of the three steps. The partition of the domain, DG operators, and related notation are introduced in \Cref{sec:NS:scheme_preliminaries}. We also use the notation $(\cdot,\cdot)$ for the L$^2$ inner-product on $\Omega$.\\
The fully discrete scheme of the velocity step \cref{eq:NS:time_dis_velocity_step} reads: for any $1 \leq n \leq N$\,, find $\vec{v}^n_h\in\IP_1(\setE_h)^3$ such that for all $\vec{\theta}_h\in\IP_1(\setE_h)^3$:
\begin{align}\label{eq:NS:space_dis_velocity_step}
\begin{split}
(\vec{v}^n_h,\vec{\theta}_h) + \tau\, a_{\mathcal{C}}(\vec{v}^{n-1}_h,\vec{v}^n_h,\vec{\theta}_h) \,+&\, \frac{\tau}{\Rey} a_\strain(\vec{v}^n_h, \vec{\theta}_h) ~=~ (\vec{v}^{n-1}_h,\vec{\theta}_h) \\
+&\, \frac{\tau}{\mathrm{X}} b_{\mathcal{P}}(p^{n-1}_h,\vec{\theta}_h) - \frac{\tau}{\mathrm{X}} (\grad{\phi^{n-1}_h},\vec{\theta}_h) + \tau b_{\mathcal{C}}(\vec{\theta}_h) + \frac{\tau}{\Rey} b_\strain(\vec{\theta}_h)~.
\end{split}
\end{align}
For the pressure projection step \cref{eq:NS:time_dis_pressure_step}, the fully-discrete scheme is: for any $1 \leq n \leq N$\,, find $\phi^n_h\in\IP_1(\setE_h)$ such that for all $\chi_h\in\IP_1(\setE_h)$:
\begin{align}\label{eq:NS:space_dis_pressure_step}
a_{\mathrm{ellip}}(\phi^n_h,\chi_h) ~=~ -\frac{\mathrm{X}}{\tau} (\div{\vec{v}^n_h},\chi_h)~.
\end{align}
For the postprocessing step \cref{eq:NS:time_dis_post_step}, the discrete pressure and velocity  are updated as follows: find $p^n_h\in\IP_1(\setE_h)$ and $\vec{u}^n_h\in\IP_1(\setE_h)^3$ such that for all  $\chi_h\in\IP_1(\setE_h)$
and $\vec{\theta}_h\in\IP_1(\setE_h)^3$:

\begin{subequations}\label{eq:NS:space_dis_post_step}
\begin{align}
(p^n_h,\chi_h) ~&=~ (p^{n-1}_h,\chi_h) + (\phi^n_h,\chi_h) - \frac{\mathrm{X}}{\Rey}(\div{\vec{v}^n_h},\chi_h)~,\label{eq:NS:space_dis_post_step_p}\\
(\vec{u}^n_h,\vec{\theta}_h) ~&=~ (\vec{v}^n_h,\vec{\theta}_h) - \frac{\tau}{\mathrm{X}}(\grad{\phi^n_h},\vec{\theta}_h)~.\label{eq:NS:space_dis_post_step_u}
\end{align}
The operator $\grad{}$ in \cref{eq:NS:space_dis_velocity_step,eq:NS:space_dis_post_step_u} denotes the broken gradient, the operator $\div{}$ in \cref{eq:NS:space_dis_pressure_step,eq:NS:space_dis_post_step_p} denotes the broken divergence. 
\begin{remark}
We note that when $\partial{\Omega}^\mathrm{in} = \emptyset$ and $\partial{\Omega}^\mathrm{out} = \emptyset$ the open boundary system \cref{eq:NS:open_boundary_model_BC} automatically degenerates to the closed boundary system with boundary condition in \cref{eq:NS:closed_boundary_model_BC}. In this case, the time discretization in \Cref{sec:NS:time_discretization} is still effective except that we end up with a pure Neumann system for \cref{eq:NS:time_dis_pressure_step}, which means the linear system generated by \cref{eq:NS:space_dis_pressure_step} will be singular. To overcome this challenge, the constraint in \cref{eq:NS:close_system_mean_p} should be considered. We can introduce a Lagrange multiplier to produce an optimization problem and solve it. We omit technical details here, since this case is not our focus, and for those who are interested in this methodology, we refer to \cite{bochev2005finite}. 
\end{remark}

\paragraph{Div--div correction}
It is worth to note that  the standard algorithm may not provide a pointwise divergence free velocity field. A classical remedy is to apply the div--div stabilization technique \cite{krank2017high}. The projection is stabilized by adding an additional term in \cref{eq:NS:space_dis_post_step_u}. We employ the following form:
\begin{align}\label{eq:NS:space_post_improved_v}
(\vec{u}^n_h,\vec{\theta}_h) 
+ \sigma_\mathrm{div} (\div{\vec{u}^n_h},\div{\vec{\theta}_h}) 
~=~ (\vec{v}^n_h,\vec{\theta}_h) - \frac{\tau}{\mathrm{X}}\,(\grad{\phi^n_h},\vec{\theta}_h)~,
\end{align} 
where $\sigma_\mathrm{div}$ denotes a penalty coefficient. The authors in \cite{krank2017high} provide a guideline on how to select the value of this coefficient. 
Comparing with the standard algorithm, applying the divergence correction according to~\cref{eq:NS:space_post_improved_v} gives acceptable results with quite small pointwise divergence in a~computational efficient way.  \Cref{Fig:NS:div_correction} shows an example of the impact of the div--div correction technique.  The porous medium is obtained by packing spheres, which is a popular method for approximating and studying porous media. A more detailed description of the sphere pack is given in \Cref{sec:NS:numerical_simulation}.  We take a shapshot of a sphere pack simulation at the second time step. \Cref{Fig:NS:div_correction} shows the velocity magnitude and the pointwise divergence field in both cases: with and without div--div correction.  Without div--div stabilization, the maximum pointwise divergence is $10^6$ times larger than the maximum pointwise divergence for the stabilized case.  


Note, for the case of closed boundaries (cf.~\Cref{sec:Boundary_conditions}), \cref{eq:NS:space_post_improved_v} may be additionally supplemented by a~jump-penalty term~\cite{akbas2017analogue}:
\begin{align*}
(\vec{u}^n_h,\vec{\theta}_h) 
+ \sigma_\mathrm{div} (\div{\vec{u}^n_h},\div{\vec{\theta}_h})  
+ \frac{\sigma_\mathrm{cont}}{h} \sum_{e \in \Gammah \cup \partial{\Omega}} \int_e \jump{\vec{u}^n_h\cdot\normal_e}\jump{\vec{\theta}_h\cdot\normal_e}
~=~ (\vec{v}^n_h,\vec{\theta}_h) - \frac{\tau}{\mathrm{X}}\,(\grad{\phi^n_h},\vec{\theta}_h)~.
\end{align*}
with jump-penalty parameter~$\sigma_\mathrm{cont}$.  However, the above technique increases the computational cost, especially for large-scale simulations. Since the div--div correction \cref{eq:NS:space_post_improved_v} satisfies the requirements for the considered pore-scale application, we use it for the numerical simulations in \Cref{sec:NS:numerical_simulation}. 
\end{subequations}
\begin{figure}[ht!]
\centering
\begin{tabular}{@{}c@{\qquad}c@{~~}c@{}}
\includegraphics[width=0.23\textwidth]{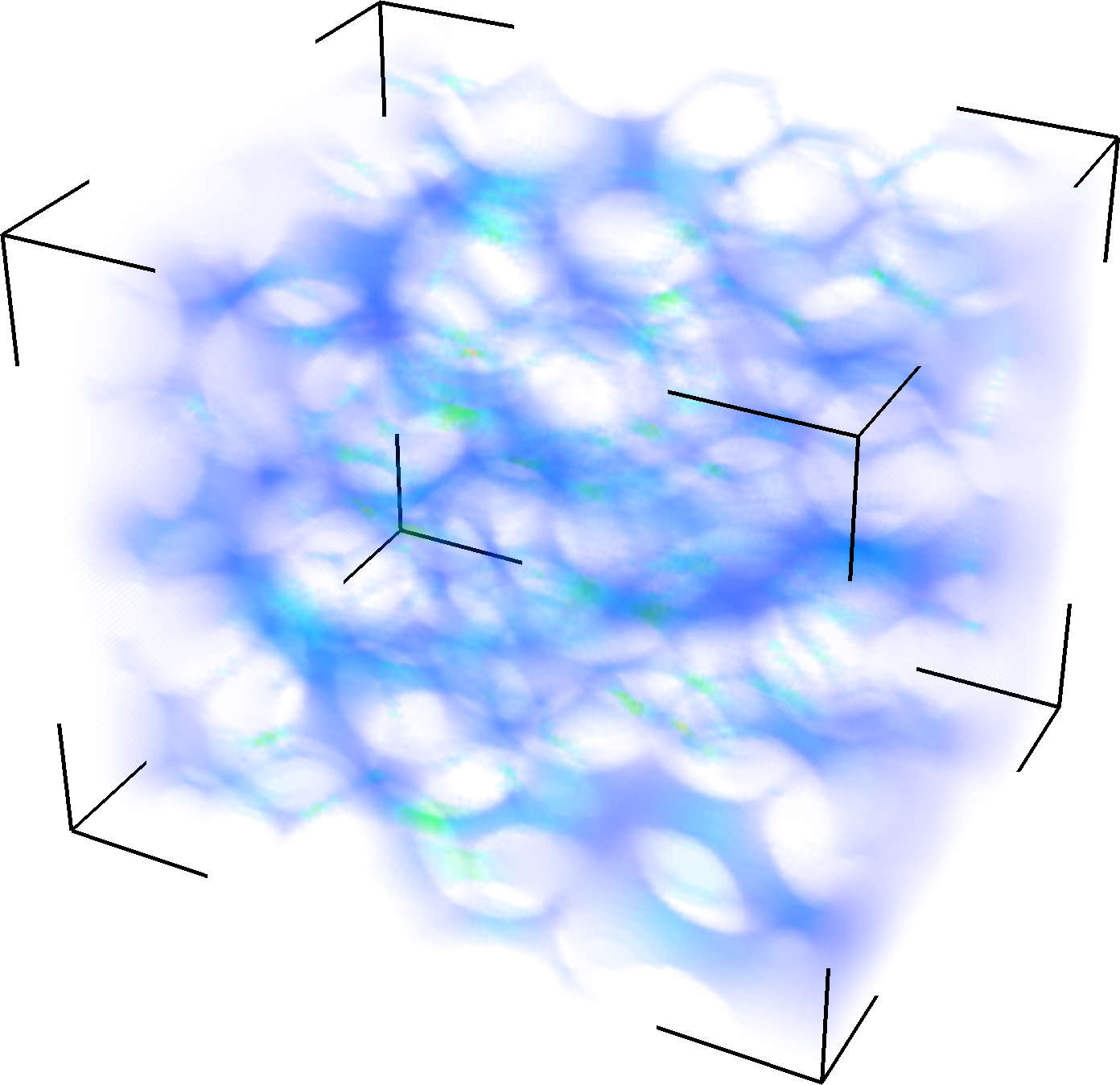}&
\includegraphics[width=0.231\textwidth]{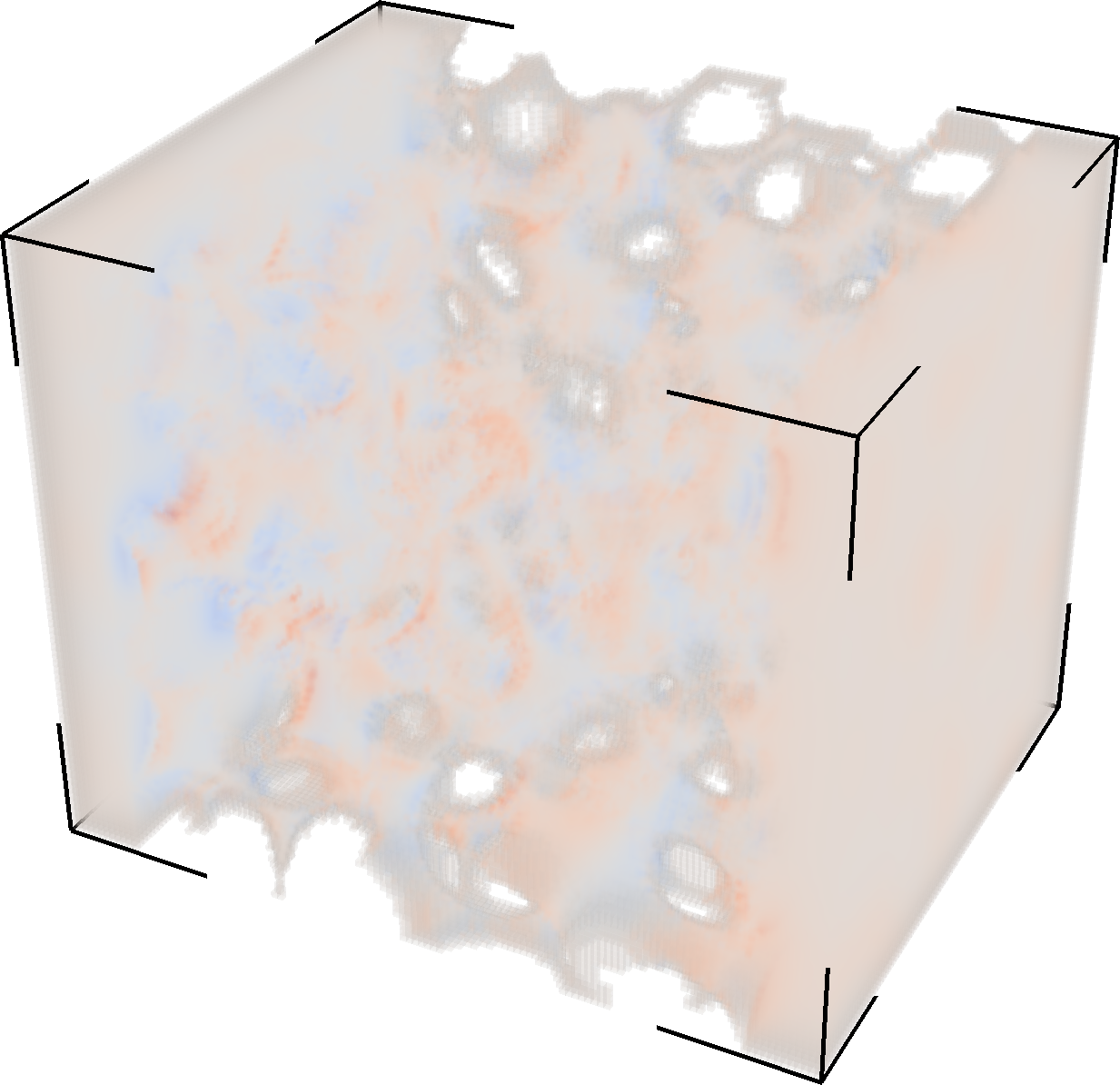}&
\includegraphics[width=0.1\textwidth]{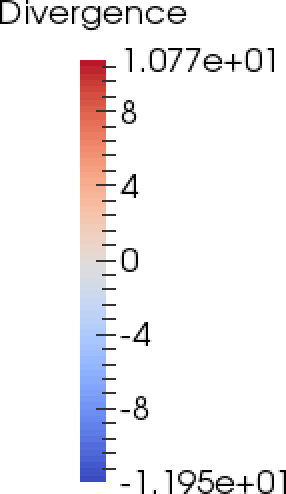}\\
(a) & (b) & \\
\includegraphics[width=0.23\textwidth]{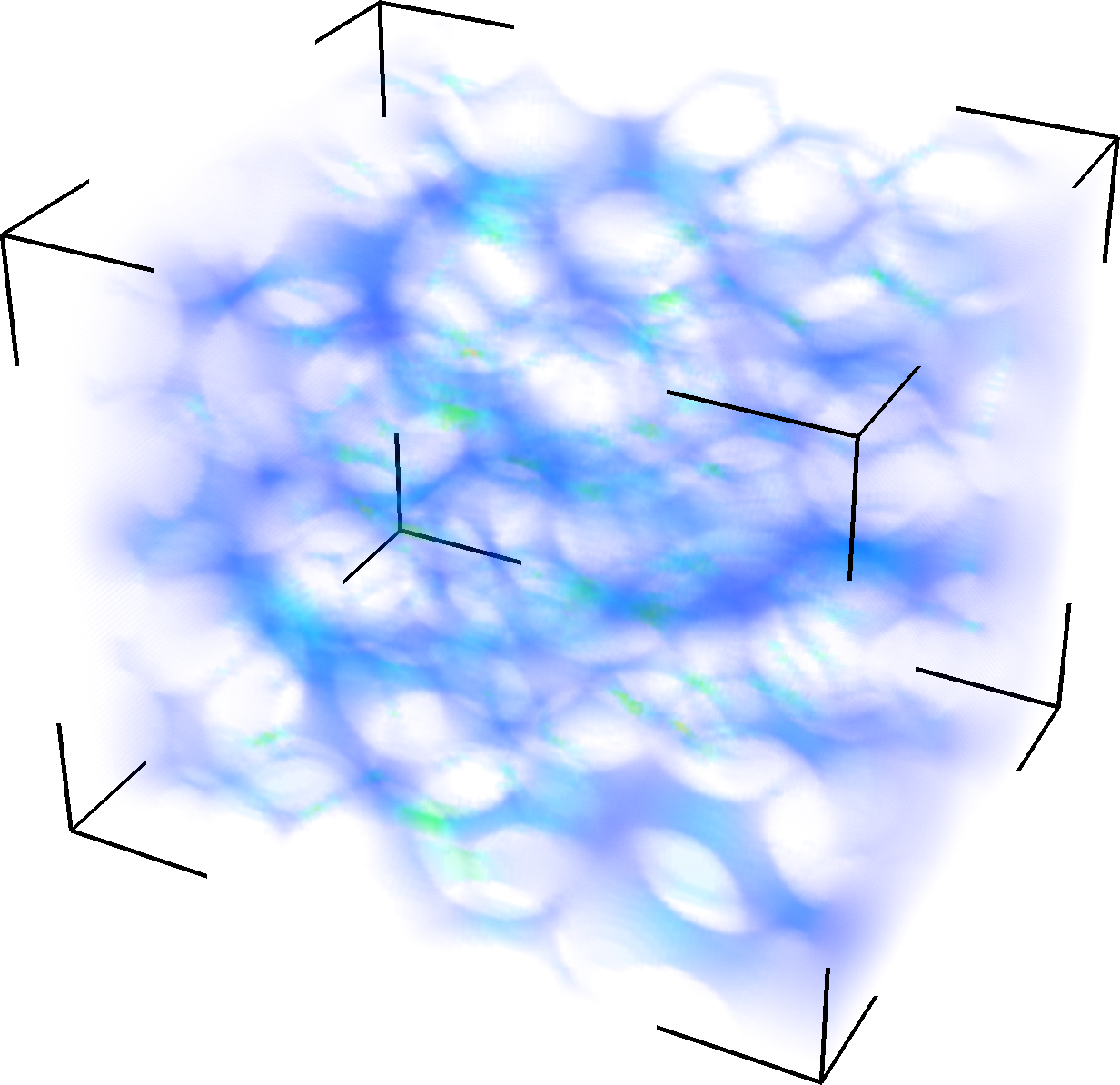}&
\includegraphics[width=0.231\textwidth]{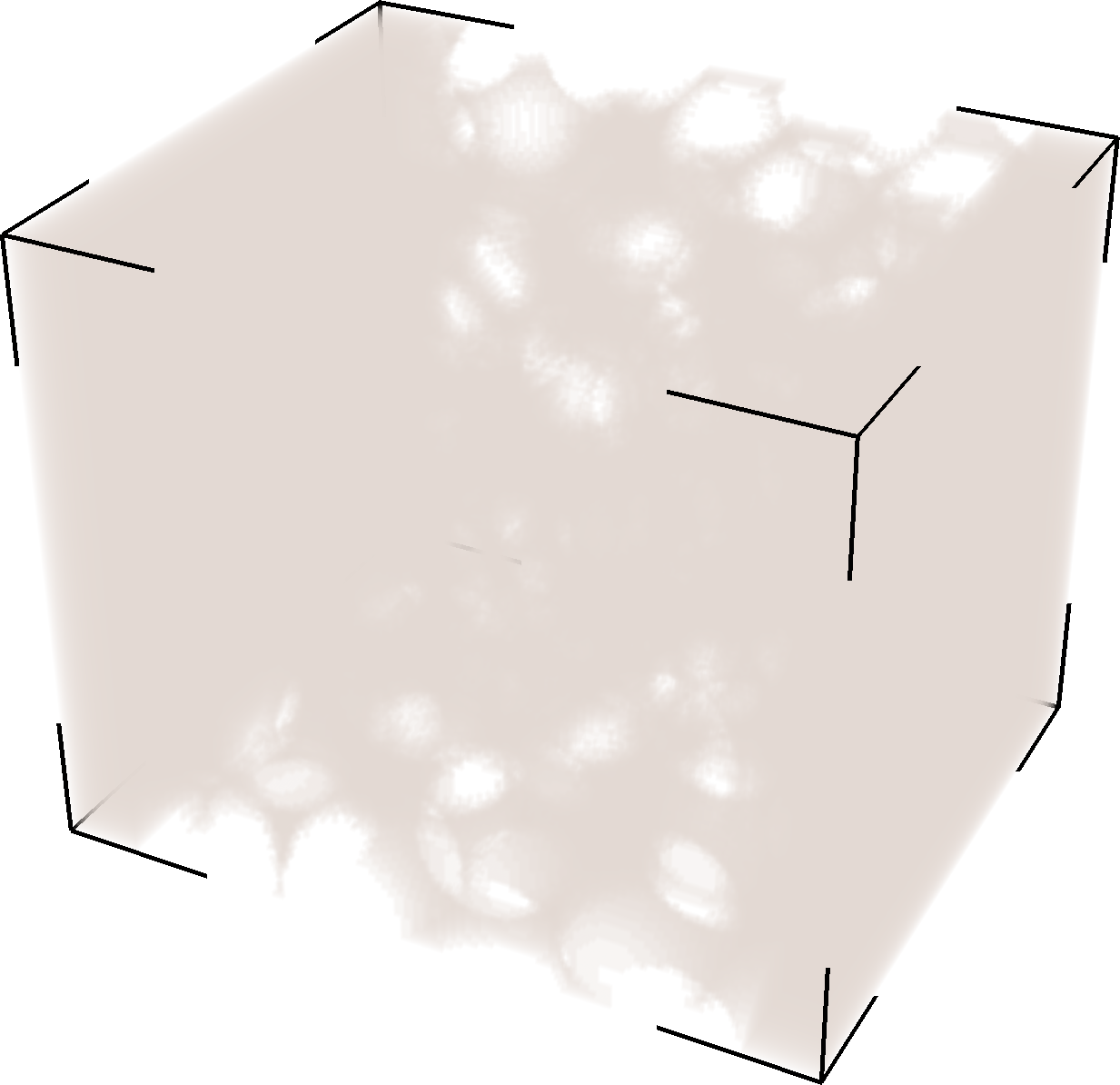}&
\includegraphics[width=0.085\textwidth]{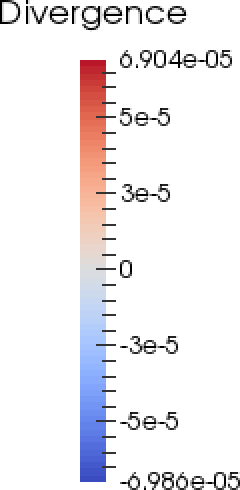}\\
(c) & (d) & 
\end{tabular}
\caption{Velocity field and divergence of velocity field in a sphere-pack. The magnitude of velocity and the divergence of
the velocity are shown in (a) and (b) for the simulations without div--div stabilization, and they are shown in (c) and (d) for the simulations with the div-div correction. The color bars are scaled to the minimum and maximum values of the divergence in the porous domain.}
\label{Fig:NS:div_correction}
\end{figure}

\section{Numerical Simulation}\label{sec:NS:numerical_simulation}
In this section, our numerical scheme is validated with several benchmark problems: Poiseuille flow in a cylinder and permeability estimation in a series of pipes with different shapes. Finally we compute the permeability in a sphere pack and a Berea rock sample.  We choose NIPG for space discretization in all of the numerical experiments. 

\subsection{Implementation and solver}
We implement the numerical algorithm on a cubic mesh, i.\,e., the domain $\Omega$ is partitioned into $\Nel$ identical regular hexahedrons. One computational advantage of selecting cubic meshes is that we do not need to store the topology of the grid and its related data, such as the indices of adjacent elements or geometrical properties of every element. Instead, these data are computed on the fly and thus no additional memory is consumed. Another factor that led to choose cubic meshes is that each element represents a voxel in the reconstructed  micro-CT images of the porous structure.
\par
Let the cube $\hat{E}=[-1,\,1]^3$ be the reference element. We employ a hierarchical modal $L^2$ orthonormal basis $\{\hat{\psi}_{j}\}$ on $\hat{E}$, which is constructed by tensor products of one dimensional Legendre polynomials. Then all the basis on element $E_k$, $0 \leq k \leq \Nel-1$, are generated by applying invertible linear transformation $\vec{F}_k$, i.\,e., the basis functions: $\psi_{kj}=\hat{\psi}_{j}\circ\vec{F}_k^{-1}$. 
There are four degrees of freedom per element.
For the approximation of surface and volume integrals we use Gauss--Legendre quadrature rules and derive rules for two and three space dimensions by tensor products. More details related with the domain triangulation and basis functions are referred to \cite{FLAR2018AdvCH}.
\par
The numerical solutions $p_h^n, \phi_h^n \in \IP_1(\setE_h)$ and $\vec{v}_h^n, \vec{u}_h^n \in \IP_1(\setE_h)^3$ can be written as linear combinations of the basis functions:
\begin{align*}
p_h^n(\vec{x}) ~&=~ \sum_{k=0}^{\Nel-1} \sum_{j=0}^{3} p_{kj}^n\,\psi _{kj}(\vec{x})~, &
\vec{v}_h^n(\vec{x}) ~&=~ \sum_{\ell=0}^{2} \sum_{k=0}^{\Nel-1} \sum_{j=0}^{3} \vec{v}_{\ell kj}^n\,\vec{e}_\ell\psi _{kj}(\vec{x})~, \\
\phi_h^n(\vec{x}) ~&=~ \sum_{k=0}^{\Nel-1} \sum_{j=0}^{3} \phi_{kj}^n\,\psi _{kj}(\vec{x})~, &
\vec{u}_h^n(\vec{x}) ~&=~ \sum_{\ell=0}^{2} \sum_{k=0}^{\Nel-1} \sum_{j=0}^{3} \vec{u}_{\ell kj}^n\,\vec{e}_\ell\psi _{kj}(\vec{x})~.
\end{align*}
The unit vectors corresponding to the Euclidean basis are denoted by $\vec{e}_\ell$, for $0\leq \ell\leq 2$.
The entry of a matrix with row index $4  k' + j'$ (resp. $4 \Nel \ell' + 4  k' + j'$) and column index $4  k + j$ (resp. $4 \Nel \ell + 4  k + j$) is denoted by $[\,\cdot\,]_{k',j';k,j}$ (resp. $[\,\cdot\,]_{\ell',k',j';\ell,k,j}$).  Due to the orthogonality of basis functions, the mass matrix is the identity matrix multiplied by a scalar. We define the following matrices and vectors:
\begin{align*}
\vecc{M} ~&=~ \frac{h^3}{8}\vecc{I}~,&
[\vec{B}_\strain^n]_{\ell',k',j'} ~&=~ b_{\strain}(\vec{e}_{\ell'}\psi_{k'j'})~,\\
[\vecc{A}_\strain]_{\ell',k',j';\ell,k,j} ~&=~ a_\strain(\vec{e}_\ell\psi_{kj},\vec{e}_{\ell'}\psi_{k'j'})~,&
[\vec{B}_\mathcal{C}^n]_{\ell',k',j'} ~&=~ b_\mathcal{C}(\vec{e}_{\ell'}\psi_{k'j'})~,\\
[\vecc{A}_\mathcal{C}^n]_{\ell',k',j';\ell,k,j} ~&=~ a_\mathcal{C}(\vec{v}_h^{n-1},\vec{e}_\ell\psi_{kj},\vec{e}_{\ell'}\psi_{k'j'})~,&
[\vec{B}_\mathcal{P}^n]_{\ell',k',j'} ~&=~ b_\mathcal{P}(p_h^n,\vec{e}_{\ell'}\psi_{k'j'})~,\\
[\vecc{A}_\varphi]_{k',j';k,j} ~&=~ a_\varphi(\psi_{kj},\psi_{k'j'})~,&
[\vec{B}_\phi^n]_{\ell',k',j'} ~&=~ -(\grad{\phi}_h^n,\vec{e}_{\ell'}\psi_{k'j'})~,\\
[\vecc{A}_\mathrm{div}]_{\ell'k'j';\ell kj} ~&=~ \big(\div{(\vec{e}_\ell\psi_{kj})},\div{(\vec{e}_{\ell'}\psi_{k'j'})}\big)~,&
[\vec{B}_\varphi^n]_{k',j'} ~&=~ -(\div{\vec{v}}_h^n,\psi_{k'j'})~.
\end{align*}
Then the matrix formulation of the velocity step reads: for any $1 \leq n \leq \Nel$, find $\vec{X}_\vec{v}^n$, where $[\vec{X}_\vec{v}^n]_{\ell,k,j} = \vec{v}_{\ell kj}^n$, such that:
\begin{align*}
(\vecc{M} + \tau\vecc{A}_\mathcal{C}^n + \frac{\tau}{\Rey}\vecc{A}_\strain)\vec{X}_\vec{v}^n ~=~ \vecc{M}\vec{X}_\vec{v}^{n-1} + \frac{\tau}{\mathrm{X}}(\vec{B}_\mathcal{P}^{n-1} + \vec{B}_\phi^{n-1}) + \tau\vec{B}_\mathcal{C}^n + \frac{\tau}{\Rey}\vec{B}_\strain^n~.
\end{align*}
For the pressure projection step, the matrix formulation of the elliptic problem reads: for any $1 \leq n \leq \Nel$, find $\vec{X}_\phi^n$, where $[\vec{X}_\phi^n]_{k,j} = \phi_{kj}^n$, such that
\begin{align*}
\vecc{A}_\varphi\vec{X}_\phi^n ~=~ \frac{\mathrm{X}}{\tau} \vec{B}_\varphi^n~.
\end{align*}
For the postprocessing step, the matrix formulation reads: for any $1 \leq n \leq \Nel$, find $\vec{X}_p^n$ and $\vec{X}_\vec{u}^n$, where $[\vec{X}_p^n]_{k,j} = p_{kj}^n$ and $[\vec{X}_\vec{u}^n]_{l,k,j} = \vec{u}_{lkj}^n$, such that
\begin{align*}
\vecc{M}\vec{X}_p^n ~&=~ \vecc{M}(\vec{X}_p^{n-1} + \vec{X}_\phi^n) + \frac{\mathrm{X}}{\Rey}\,\vec{B}_\varphi^n~,\\
(\vecc{M} + \sigma_\mathrm{div}\vecc{A}_\mathrm{div})\vec{X}_\vec{u}^n ~&=~ \vecc{M}\vec{X}_\vec{v}^n + \frac{\tau}{\mathrm{X}}\,\vec{B}_\phi^n~.
\end{align*}
In the numerical simulation, we use BiCGStab with multicolored Gauss--Seidel preconditioner to solve the linear systems above.

\subsection{Poiseuille flow in a cylinder}
The goal of this section is to validate the numerical scheme via the Poiseuille flow test, which is a simple case generally used to verify the correct implementation of the numerical algorithm. 
Let us consider the steady laminar flow of an incompressible viscous fluid through a cylindrical pipe, which length equals $L$, radius equals $R$, and cross-sectional centered at $(a,b)$: 
\begin{align}\label{eq:NS:Poiseuille_cylinder}
\Omega ~=~ \big\{(x_0,x_1,x_2)\in\IR^3:~(x_0-a)^2 + (x_1-b)^2 < R^2 ~~\text{and}~~ 0 < x_2 < L\big\}~.
\end{align}
The inflow boundary $\partial\Omega^\mathrm{in}$ is the horizontal cross-section of the cylinder in the plane $x_2 = 0$ and the outflow boundary is the horizontal cross-section in the plane $x_2 = L$.  In this set-up the analytical solution of the resulting Poiseuille flow is defined by:
\begin{align*}
\vec{v}(x_0,x_1,x_2) ~&=~ \transpose{\begin{bmatrix}0 & 0 & {\displaystyle 1-\Big(\frac{x_0-a}{R}\Big)^2-\Big(\frac{x_1-b}{R}\Big)^2}\end{bmatrix}}~,\\
p(x_0,x_1,x_2) ~&=~ \frac{4\,\mathrm{X}}{\Rey\,R^2}\,(L - x_2)~.
\end{align*}
In numerical validation test, we select a cylindrical pipe parameterized by $L=1$, $R=0.5$, and $(a,b)=(0.5,0.5)$. The resolution of the mesh to approximate this pipe is $h=1/64$. We choose $\Rey = 1$ and $\mathrm{X}=1$. 
\Cref{fig:NS:Poiseuille_flow_v_p} shows the profiles of the velocity and pressure along the line defined by $x_1=x_2=0.5$ and the line defined by $x_0=x_1=0.5$ respectively.
The pointwise relative error between numerical and analytical solutions is less than $0.35\%$ for the velocity and less than
$0.57\%$ for the pressure.
\begin{figure}[ht!]\label{fig:NS:Poiseuille_velocity_and_pressure}
\centering
\includegraphics[width=0.885\textwidth]{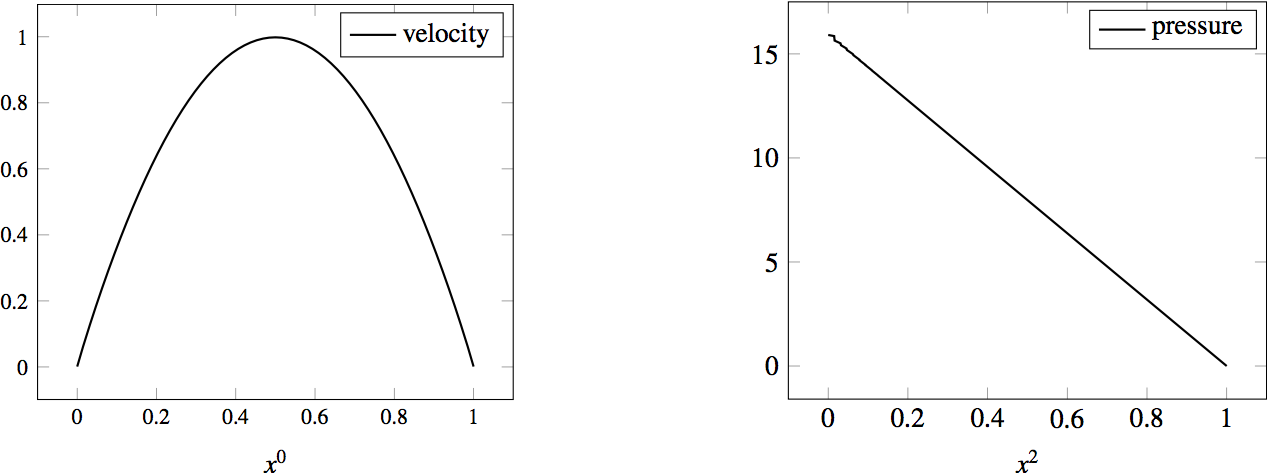}
\caption{Velocity magnitude along line $\{x_1=0.5,\,x_2=0.5\}$ (left) and pressure magnitude along line $\{x_0=0.5,\,x_1=0.5\}$ (right).}
\label{fig:NS:Poiseuille_flow_v_p}
\end{figure}

\subsection{Convergence study}
In this section, we verify our pressure-projection algorithm numerically by obtaining convergence rates with the method of manufactured solutions. We obtain temporal rates of convergence by computing the time-dependent nonlinear smooth solutions with decreasing time steps. The prescribed analytical solutions on a cubic domain $\Omega = (0,1)^3$ are as follows:
\begin{align*}
v_x(x,y,z,t) ~&=~ -\exp{(-t+x)}\sin{(y+z)}-\exp{(-t+z)}\cos{(x+y)}~,\\
v_y(x,y,z,t) ~&=~ -\exp{(-t+y)}\sin{(x+z)}-\exp{(-t+x)}\cos{(y+z)}~,\\
v_z(x,y,z,t) ~&=~ -\exp{(-t+z)}\sin{(x+y)}-\exp{(-t+y)}\cos{(x+z)}~,\\
p(x,y,z,t)   ~&=~ -\exp{(-2t)}\Big(\exp(x+z)\sin{(y+z)}\cos{(x+y)}\\
&\hspace{2cm}+\exp(x+y)\sin{(x+z)}\cos{(y+z)}\\
&\hspace{2cm}+\exp(y+z)\sin{(x+y)}\cos{(x+z)}\\
&\hspace{2cm}+\frac{1}{2}\exp{(2x)} + \frac{1}{2}\exp{(2y)} + \frac{1}{2}\exp{(2z)}
-7.63958172715414\Big)~.
\end{align*}
Note that the above solutions belong to a Beltrami flow family \cite{ethier1994exact},
with  parameters $\rho_0 = 1$ and $\mu_\mathrm{s} = 1$, and satisfy the incompressible Navier--Stokes system \cref{eq:GoverningEquationsNS} with average pressure zero condition \cref{eq:NS:close_system_mean_p} up to machine precision for any $t \in [0,T]$. 
\par
In this convergence study, the mesh resolution is set to be small enough, $h=1/100$, to guarantee the error in space is significantly smaller than the splitting error and we employ pure Dirichlet boundary condition on $\partial{\Omega}$. Taking the simulation end time $T=1.0$, the convergence rates in $L^2$ norm are reported in \Cref{tab:NS:convergence_rate_study}, which matches the expected first order convergence in time. We remark that the spatial convergence of the NIPG method with piecewise linears is second order with respect to the $L^2$ norm, which can be tested by prescribing a time-independent solution such as the Poiseuille flow in a cylindrical domain and run the simulator to steady state. 
\begin{table}[ht!]
\centering
\begin{tabularx}{\linewidth}{@{}C@{}|C@{}C@{}|C@{}C@{}|C@{}C@{}}
\toprule
$\tau$ & $\norm{\vec{v}_h^N-\vec{v}(T)}{L^2}$ & rate & $\norm{\vec{u}_h^N-\vec{v}(T)}{L^2}$ & rate & $\norm{p_h^N-p(T)}{L^2}$ & rate\\
\midrule
$1/2^3$ & $1.525$\,E$-2$ & ---     & $7.182$\,E$-3$ & ---     & $2.201$\,E$-1$ & ---    \\
$1/2^4$ & $5.726$\,E$-3$ & $1.413$ & $2.796$\,E$-3$ & $1.361$ & $1.094$\,E$-1$ & $1.009$\\
$1/2^5$ & $1.604$\,E$-3$ & $1.836$ & $7.900$\,E$-4$ & $1.824$ & $3.685$\,E$-2$ & $1.570$\\
$1/2^6$ & $4.521$\,E$-4$ & $1.827$ & $2.286$\,E$-4$ & $1.789$ & $1.265$\,E$-2$ & $1.543$\\
$1/2^7$ & $1.317$\,E$-4$ & $1.780$ & $7.507$\,E$-5$ & $1.607$ & $5.004$\,E$-3$ & $1.338$\\
$1/2^8$ & $4.635$\,E$-5$ & $1.507$ & $3.661$\,E$-5$ & $1.036$ & $2.399$\,E$-3$ & $1.061$\\
\bottomrule
\end{tabularx}
\caption{Errors and convergence rates of velocity and pressure.}
\label{tab:NS:convergence_rate_study}
\end{table}

\subsection{Permeability estimation in pipes}
Permeability is an important characteristic of porous media, that measures the ability of a porous medium to allow fluids to pass through. The most widely used unit for permeability is Darcy~$[\mathrm{D}]$ or milli-Darcy~$[\mathrm{mD}]$, where $1\,\mathrm{D}=10^3\,\mathrm{mD}$ and $1\,\mathrm{mD}=9.869\times10^{-16}\,\mathrm{m}^2$. High permeability of a material indicates fluid easily flows through. 
\par
The basic principle for estimating permeability is Darcy's law, which expresses the relation between permeability~$k$ of a porous medium block, cross section area $A$ and length $L$ as follows
\begin{equation*}
k ~=~ \frac{\mu_\mathrm{s}\,L\,Q}{A\, (P_\mathrm{in} - P_\mathrm{out})}~,
\end{equation*}
where $Q$ denotes the volumetric flux and $(P_\mathrm{in} - P_\mathrm{out})$ denotes the pressure difference between inlet and outlet. Thus, we compute the average flux and pressure from its numerical approximations by integration over faces:
\begin{align*}
Q ~&=~ \frac{1}{\abs{S_\mathrm{out}}} \int_{S_\mathrm{out}} \vec{u}_h^N\cdot\normal_\mathrm{main}~,\\
P_\mathrm{in} ~&=~ \frac{1}{\abs{S_\mathrm{in}}} \int_{S_\mathrm{in}} p_h^N, \quad
P_\mathrm{out} ~=~ \frac{1}{\abs{S_\mathrm{out}}} \int_{S_\mathrm{out}} p_h^N~.
\end{align*}
Here $S_\mathrm{in}$ and $S_\mathrm{out}$ denote sample faces near the inlet and outlet  respectively (at a distance of 10\% inside the domain), and $\normal_\mathrm{main}$ is the main flow direction, i.\,e., a unit vector pointing from $S_\mathrm{in}$ to $S_\mathrm{out}$. The reason that we take sample faces near the in/outlet is to minimize the influence of boundary conditions on in/outflow surfaces. Note in all of the following numerical experiments, we take the parameters $\rho_0 = 10^3\,\mathrm{kg}\,\mathrm{m}^{-3}$ and $\mu_\mathrm{s} = 10^{-3}\,\mathrm{kg}\,\mathrm{m}^{-1}\,\mathrm{s}^{-1}$, which approximate the state of water at $293.15\,\mathrm{K}$ ($20\,^{\circ}\text{C}$). 
\begin{figure}[ht!]
\centering
\includegraphics[width=\textwidth]{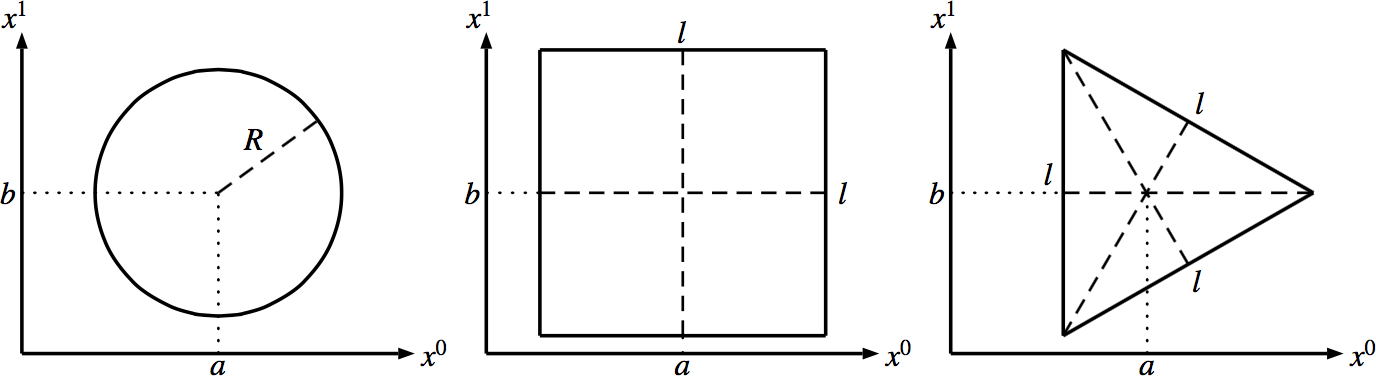}
\caption{Pipes have a~length of~$L$ in~$x_2$-direction and the following cross-sections: circle (left), square (center), and equilateral-triangle (right).}
\label{fig:pipe_geometry}
\end{figure}
In this part, we estimate the permeability of pipes with different shapes, which include the circled, squared, and equilateral-triangle scenarios, see \Cref{fig:pipe_geometry}. The geometry of a cylindrical pipe is defined in \cref{eq:NS:Poiseuille_cylinder}. The geometry of a squared pipe, which length equals~$L$, edge length equals $\ell$, and cross-sectional centered at $(a,\,b)$, in Cartesian coordinates is defined by
\begin{align*}
\Omega ~=~ \big\{(x_0,x_1,x_2)\in\IR^3:~ \abs{x_0-a} < \frac{\ell}{2},~~ \abs{x_1-b} < \frac{\ell}{2}, ~~\text{and}~~ 0 < x_2 < L\big\}~.
\end{align*} 
The geometry of an equilateral-triangle pipe, which length equals~$L$, side length equals $\ell$, and cross-sectional centered at $(a,\,b)$, in Cartesian coordinate is defined by
\begin{align*}
\Omega ~=~ \big\{(x_0,x_1,x_2)\in\IR^3:~ x_0 - a +\frac{\sqrt{3}}{6}\ell > 0,~~ x_0 + \sqrt{3}x_1 - a - \sqrt{3}b - \frac{\sqrt{3}}{3}\ell &< 0,\\ x_0 - \sqrt{3}x_1 - a + \sqrt{3}b - \frac{\sqrt{3}}{3}\ell &< 0, ~~\text{and}~~ 0 < x_2 < L\big\}~.
\end{align*}
The analytical formulae for calculating the permeability of a block, cross section area $A$ with a pipe hole in the middle are \cite{saxena2017references}
\begin{align*}
\text{cylindrical pipe:} && k ~&=~ \frac{\pi R^4}{8A}~,\\
\text{squared pipe:} && k ~&=~ 0.035144\,\frac{\ell^4}{A}~,\\
\text{equilateral-triangle pipe:} && k ~&=~ \frac{\sqrt{3}\,\ell^4}{320\,A}~.
\end{align*}
It is worth noting that although the shape of the pipe is represented by a specific analytical expression, the permeability value of a particular pipe is independent of coordinates. We fix the block's cross section area $A = 1\,\mathrm{m}^2$ and use above formulae to generate a series of theoretical values as standard references, see the second column in \Cref{tab:NS:permeability_cylindrical_pipe}, \Cref{tab:NS:permeability_squared_pipe}, and \Cref{tab:NS:permeability_triangular_pipe}. A well defined inflow boundary condition plays an important role in computing accurate permeability \cite{lekner2007viscous}. We prescribe the following velocity distribution $\vec{v}^\mathrm{in} = \transpose{(0,\,0,\,v_z)}$ on inflow boundary, where $v_z$ is expressed as
\begin{align*}
\text{cylindrical pipe:} && v_z ~&=~ 1-\Big(\frac{x_0-a}{R}\Big)^2-\Big(\frac{x_1-b}{R}\Big)^2~,\\
\text{squared pipe:} && v_z ~&=~ \frac{16}{\ell^4}\,\Big(x_0 - a + \frac{\ell}{2}\Big)\Big(x_0 - a - \frac{\ell}{2}\Big)\Big(x_1 - b + \frac{\ell}{2}\Big)\Big(x_1 - b - \frac{\ell}{2}\Big)~,\\
\text{equilateral-triangle pipe:} && v_z ~&=~ \frac{6\sqrt{3}}{\ell^3}\,\Big(x_0 - a + \frac{\sqrt{3}\ell}{6}\Big)\Big(\big(x_0 - a)^2 - 3(x_1 - b - \frac{\ell}{3}\big)^2\Big)~.
\end{align*}
The velocity profiles above are shown in \Cref{fig:NS:pipes_inflow_velocity} and for the motivation of using these velocity distribution on inflow boundary we refer to \cite{lekner2007viscous}. Once the numerical flow reaches steady-state, we estimate the permeabilities for different sizes of pipes. Results are  shown in the third column of \Cref{tab:NS:permeability_cylindrical_pipe}, \Cref{tab:NS:permeability_squared_pipe}, and \Cref{tab:NS:permeability_triangular_pipe}. The last column of these tables provide the relative error.  We observe that our numerical scheme is accurate and that all relative errors are in the range  $[0.18\%, 5.25\%]$ for the different shape sizes.
\begin{figure}[ht!]
\centering
\begin{tabularx}{\linewidth}{@{}C@{}C@{}C@{}}
\includegraphics[width=0.295\textwidth]{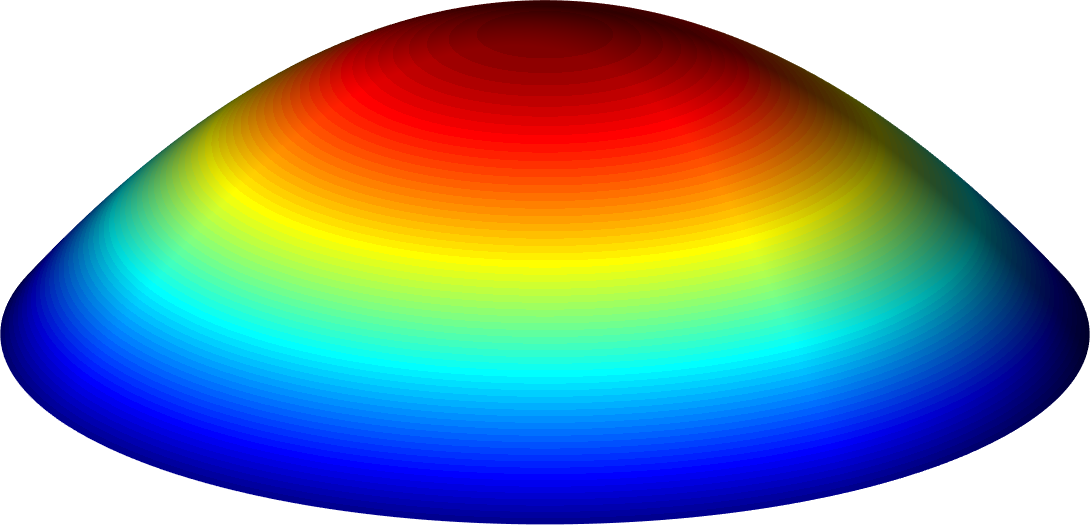}&
\includegraphics[width=0.32\textwidth]{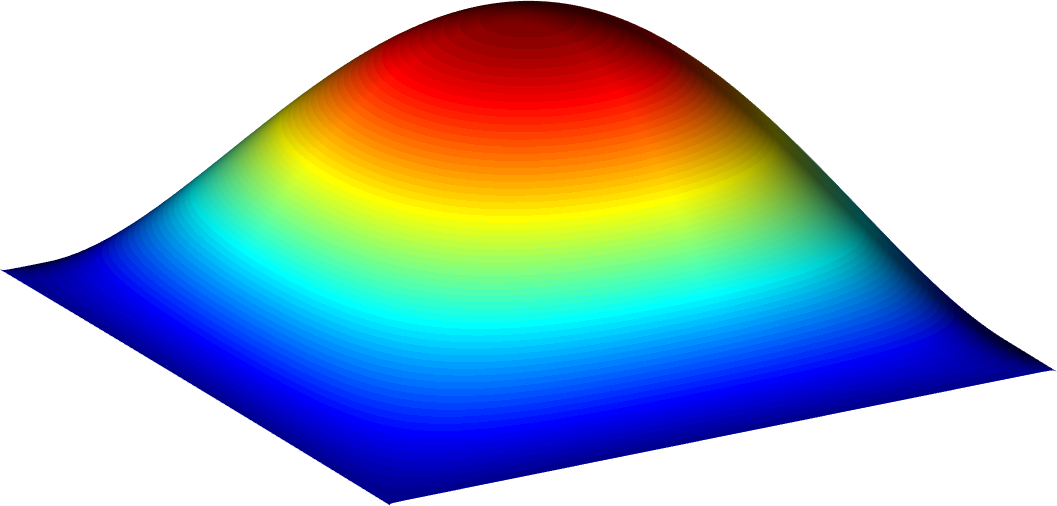}&
\includegraphics[width=0.30\textwidth]{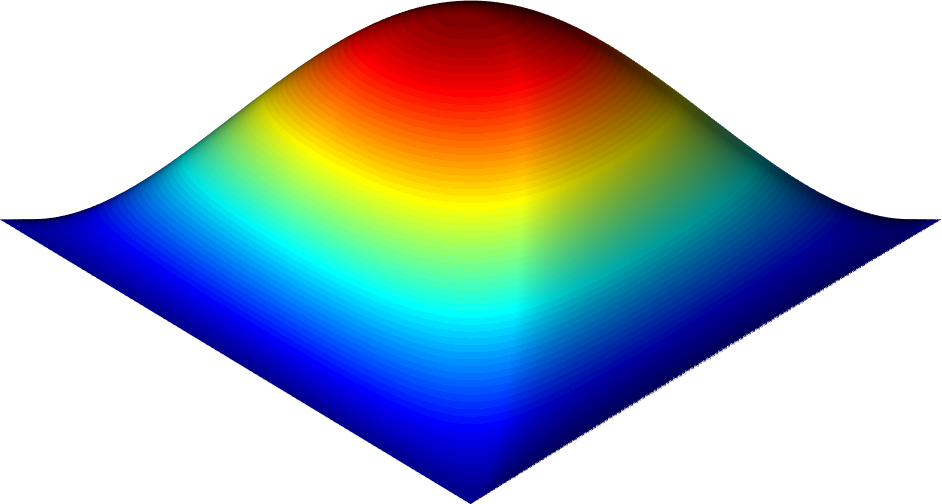}
\end{tabularx}
\caption{Velocity distributions on inflow boundary: cylinder (left), square (center), and equilateral-triangle (right).}
\label{fig:NS:pipes_inflow_velocity}
\end{figure}
\begin{table}[ht!]
\centering
\begin{tabularx}{\linewidth}{>{\centering}m{2.3cm}|C@{}C@{}|c@{}}
\toprule
radius $[\mathrm{m}]$ & theoretical value $[\mathrm{mD}]$ & estimated value $[\mathrm{mD}]$ & relative error\\
\midrule
$2.25$\,E$-4$ & $1.020$\,E$+0$ & $9.853$\,E$-1$ & $3.40\%$ \\
$4.00$\,E$-4$ & $1.019$\,E$+1$ & $1.014$\,E$+1$ & $0.49\%$ \\
$7.00$\,E$-4$ & $9.554$\,E$+1$ & $9.670$\,E$+1$ & $1.21\%$ \\
$1.25$\,E$-3$ & $9.714$\,E$+2$ & $9.535$\,E$+2$ & $1.84\%$ \\
$2.25$\,E$-3$ & $1.020$\,E$+4$ & $9.822$\,E$+3$ & $3.71\%$ \\
$4.00$\,E$-3$ & $1.019$\,E$+5$ & $1.014$\,E$+5$ & $0.49\%$ \\
\bottomrule
\end{tabularx}
\caption{Permeability estimations on cylindrical pipes.}
\label{tab:NS:permeability_cylindrical_pipe}
\end{table}
\begin{table}[ht!]
\centering
\begin{tabularx}{\linewidth}{>{\centering}m{2.3cm}|C@{}C@{}|c@{}}
\toprule
edge length $[\mathrm{m}]$ & theoretical value $[\mathrm{mD}]$ & estimated value $[\mathrm{mD}]$ & relative error\\
\midrule
$4.09$\,E$-4$ & $9.965$\,E$-1$ & $1.036$\,E$+0$ & $3.96\%$ \\
$7.28$\,E$-4$ & $1.000$\,E$+1$ & $9.499$\,E$+0$ & $5.01\%$ \\
$1.30$\,E$-3$ & $1.017$\,E$+2$ & $9.967$\,E$+1$ & $2.00\%$ \\
$2.30$\,E$-3$ & $9.965$\,E$+2$ & $9.814$\,E$+2$ & $1.52\%$ \\
$4.09$\,E$-3$ & $9.965$\,E$+3$ & $1.034$\,E$+4$ & $3.76\%$ \\
$7.28$\,E$-3$ & $1.000$\,E$+5$ & $9.475$\,E$+4$ & $5.25\%$ \\
\bottomrule
\end{tabularx}
\caption{Permeability estimations on squared pipes.}
\label{tab:NS:permeability_squared_pipe}
\end{table}
\begin{table}[ht!]
\centering
\begin{tabularx}{\linewidth}{>{\centering}m{2.3cm}|C@{}C@{}|c@{}}
\toprule
side length $[\mathrm{m}]$ & theoretical value $[\mathrm{mD}]$ & estimated value $[\mathrm{mD}]$ & relative error\\
\midrule
$6.53$\,E$-4$ & $9.972$\,E$-1$ & $9.954$\,E$-1$ & $0.18\%$ \\
$1.17$\,E$-3$ & $1.028$\,E$+1$ & $1.067$\,E$+1$ & $3.79\%$ \\
$2.07$\,E$-3$ & $1.007$\,E$+2$ & $1.015$\,E$+2$ & $0.79\%$ \\
$3.67$\,E$-3$ & $9.949$\,E$+2$ & $1.004$\,E$+3$ & $0.91\%$ \\
$6.53$\,E$-3$ & $9.972$\,E$+3$ & $9.928$\,E$+3$ & $0.44\%$ \\
$1.17$\,E$-2$ & $1.028$\,E$+5$ & $1.066$\,E$+5$ & $3.70\%$ \\
\bottomrule
\end{tabularx}
\caption{Permeability estimations on equilateral-triangle pipes.}
\label{tab:NS:permeability_triangular_pipe}
\end{table}

\subsection{Permeability estimation in sphere pack}\label{sec:NS:sphere_pack}
The randomly non-overlapping arrangement of spheres within a hexahedral domain can be viewed as a manufactured porous structure. Studying the flow aspects in artificial channels of a spherical packing inspires researchers to understand the complex behavior of fluid flow in realistic porous media. Computing permeability of a sphere pack image can be considered as a good starter for further permeability estimations of realistic rock samples.
\par
Analogously to the lab experiment set-up, we add buffer zones before and after the inlet and outlet faces of medium. The other sides are enclosed by solid walls. More precisely, after nondimensionalization, we embed the porous image into a square cross-section channel (see \Cref{fig:buffer}) as follows:
\begin{equation*}
\Omega ~=~ \big\{(x_0,x_1,x_2)\in\IR^3:~ 0 < x_0 < 1,~~ 0.1< x_1 < 0.9, ~~\text{and}~~ 0.1< x_2 < 0.9 \big\}~.
\end{equation*}
This design allows us to prescribe the boundary conditions in a convenient way. We set the initial velocity $\vec{v}^0=\vec{0}$ and employ the following velocity distribution $\vec{v}^\mathrm{in}=\transpose{(v_x,0,0)}$ on inflow boundary, where $v_x$ is expressed as
\begin{equation*}
v_x ~=~ \frac{125}{32}\,(x_1 - 0.1)(x_1 - 0.9)(x_2 - 0.1)(x_2 - 0.9)~.
\end{equation*}
Note, as above, the main flow direction $\normal_\mathrm{main}$ equals $\transpose{(1,0,0)}$, i.\,e., the $x_0$-direction. The average pressure and flux are computed on faces $S_\mathrm{in}$ and $S_\mathrm{out}$ that are located in the middle of the inlet and outlet buffer zones, respectively.
\begin{figure}[ht!]
\centering
\includegraphics[width=0.4\linewidth]{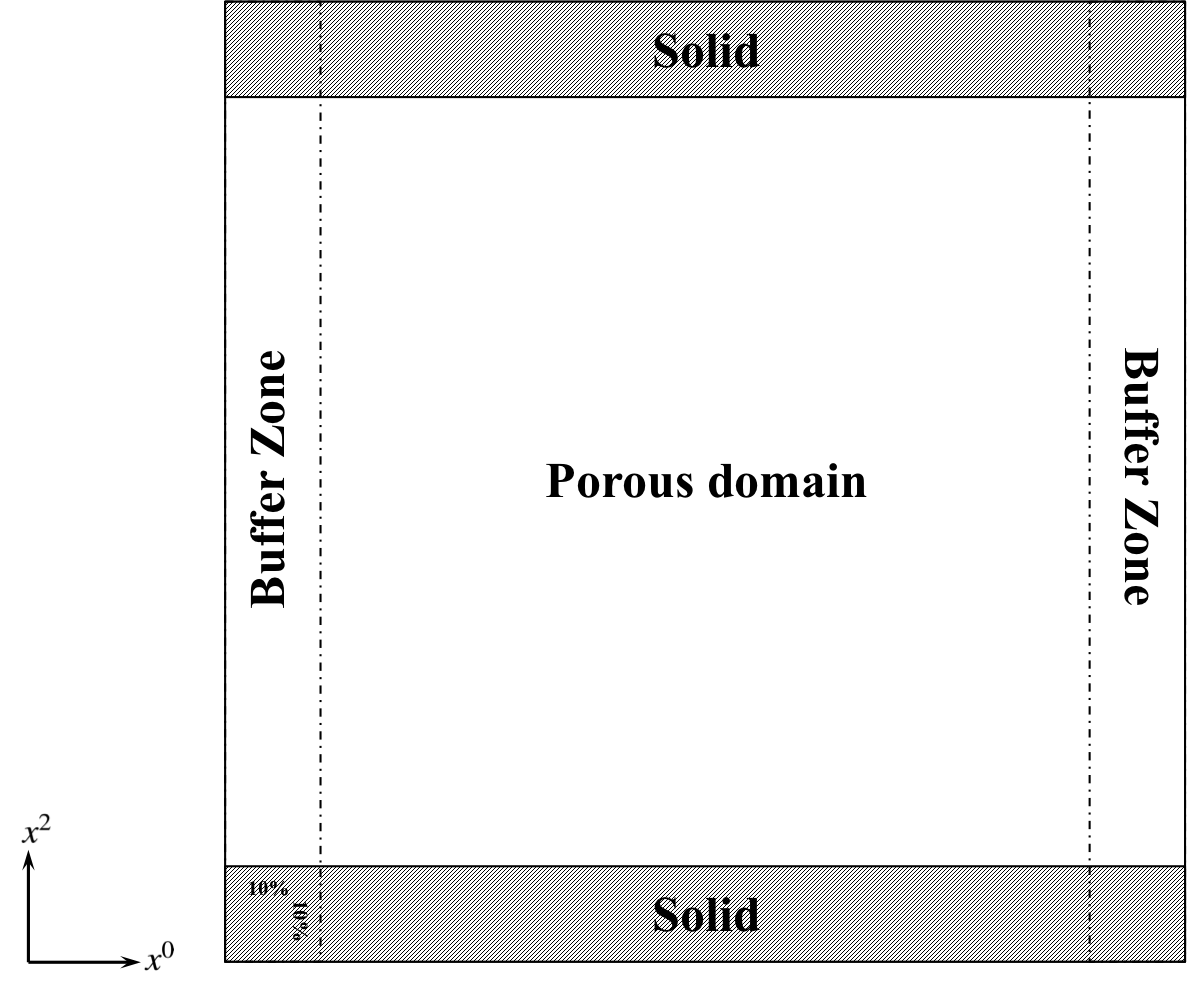}
\caption{A porous medium embedded in a square cross-section channel with buffer zones on both inlet and outlet (view in $x^1$-direction).}
\label{fig:buffer}
\end{figure}
\par
In this numerical experiment, we set the mesh resolution to $h = 1/124$. The buffers are $10\%$ on each side of the domain, i.\,e., the sample domain used for sphere pack is of size $100\times100\times100$ and faces $S_\mathrm{in}$ and $S_\mathrm{out}$ are located $5\%$ inside the domain from the inflow and outflow boundaries. We choose the characteristic length $x_\mathrm{c} = 10^{-3}\,\mathrm{m}$, the characteristic velocity $v_\mathrm{c} = 10^{-4}\,\mathrm{m}\,\mathrm{s}^{-1}$, and the characteristic pressure $p_\mathrm{c} = 1\,\mathrm{kg}\,\mathrm{m}^{-1}\,\mathrm{s}^{-2}$, which yields $\Rey = 10^{-1}$, $X = 10^{-5}$.
The magnitude of the velocity field and the pressure field at steady state are shown in \Cref{fig:NS:SpherepackNIPG_124_velocity_and_pressure}.  The computation of the permeability of this sample in the direction parallel to the flow yields the following results:
\begin{align*}
\text{permeability in~}x_0\text{-direction:} && k ~&=~ 5505\,\mathrm{mD}~,\\ 
\text{permeability in~}x_1\text{-direction:} && k ~&=~ 4027\,\mathrm{mD}~,\\
\text{permeability in~}x_2\text{-direction:} && k ~&=~ 4763\,\mathrm{mD}~. 
\end{align*}
\begin{figure}[ht!]
\centering
\includegraphics[width=\textwidth]{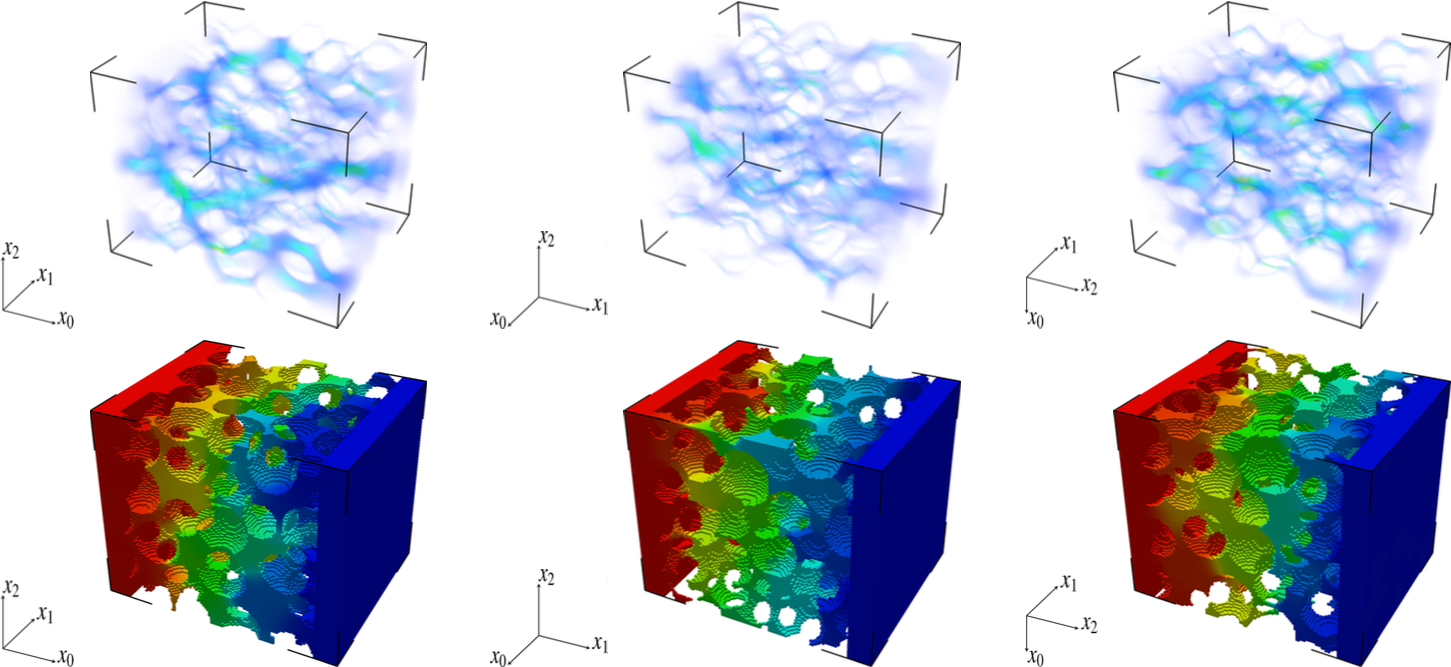}
\caption{Permeability estimation of sphere pack: magnitude of velocity field (top) and pressure field (bottom) at steady state. From left to right, the main flow directions are $x_0$-direction, $x_1$-direction, and $x_2$-direction, respectively.}
\label{fig:NS:SpherepackNIPG_124_velocity_and_pressure}
\end{figure}

\subsection{Permeability estimation in Berea sandstone}
Berea sandstone, also regarded as Berea Grit, as a type of sedimentary rock are widely distributed in nature. With suitable porosity and permeability, it has been recognized as one of the best samples in hydrogeology and it is frequently used as a source of laboratory material in petroleum industry for testing permeability. Depending on the sample, the order of permeability typically lies between $10^1\,\mathrm{mD}$ and $10^3\,\mathrm{mD}$ for a Berea sandstone core. 
\par
The porous medium used here has a resolution of $256\times256\times256$ voxels, which was created by micro-CT imaging of a Berea sandstone.
The porosity of the rock sample is $0.2$. As in the previous example with sphere-pack, we add buffer zones to the inlet and outlet faces and enclose the other boundary faces with solid rock.  The computational domain is then embedded inside a squared pipe with mesh size $h = 1/320$ and each buffer zone occupying $10\%$ of the domain. The boundary conditions and the location of faces $S_\mathrm{in}$ and $S_\mathrm{out}$ for computing average pressure and flux are the same as in \Cref{sec:NS:sphere_pack}.
\par
We choose the characteristic length $x_\mathrm{c} = 10^{-3}\,\mathrm{m}$, the characteristic velocity $v_\mathrm{c} = 10^{-4}\,\mathrm{m}\,\mathrm{s}^{-1}$, and the characteristic pressure $p_\mathrm{c} = 1\,\mathrm{kg}\,\mathrm{m}^{-1}\,\mathrm{s}^{-2}$, which yields $\Rey = 10^{-1}$ and $X = 10^{-5}$.
At steady-state, the magnitude of the velocity field and the pressure field are shown in \Cref{fig:NS:BereaNIPG_320_velocity_and_pressure}. 
The numerical solutions yield the following results for the Berea sandstone sample:
\begin{align*}
\text{permeability in~}x_0\text{-direction:} && k ~&=~ 939.4\,\mathrm{mD}~,\\ 
\text{permeability in~}x_1\text{-direction:} && k ~&=~ 1015.9\,\mathrm{mD}~,\\
\text{permeability in~}x_2\text{-direction:} && k ~&=~ 742.5\,\mathrm{mD}~. 
\end{align*}
\begin{figure}[ht!]
\centering
\includegraphics[width=\textwidth]{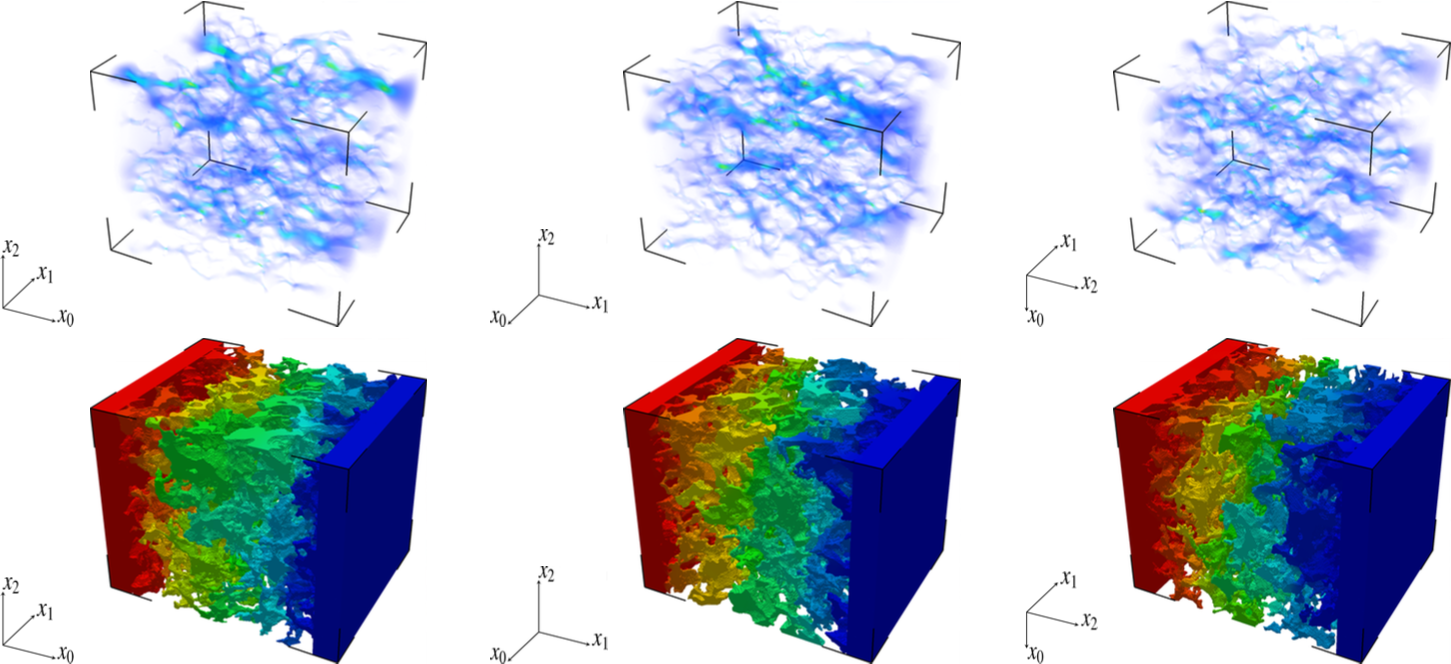}
\caption{Permeability estimation of Berea sandstone: magnitude of velocity field (top) and pressure field (bottom) at steady state. From left to right, the main flow directions are $x_0$-direction, $x_1$-direction, and $x_2$-direction, respectively.}
\label{fig:NS:BereaNIPG_320_velocity_and_pressure}
\end{figure}

\section{Conclusion}
In this paper, we propose a numerical method for solving the incompressible Navier--Stokes equations on micro-CT images of porous media. The method combines a pressure-correction algorithm with interior penalty discontinuous Galerkin. Results show the accuracy of the method for simple porous media with known analytical permeability values. Simulations on sphere-pack and Berea sandstone show the robustness of the method. 

\section*{Acknowledgments} 
The authors thank Christopher Thiele (Rice University) and Prof. Walter Chapman (Rice University) for discussions. The authors acknowledge Shell International Exploration and Production Inc. for computational resources supporting.

\appendix
\section{Diffusion Term}\label{sec:appendix}
In model problem \cref{eq:GoverningEquationsNS_ndim}, if the diffusion term $-\laplace{\vec{v}}$ is replaced by a generalized form $-2\div{\strain{(\vec{v})}}$, then for open boundary problem the DG forms \cref{eq:NS:DG_strain_a,eq:NS:DG_strain_b} should be modified as follows:
\begin{align*}
\begin{split}
a_\strain(\vec{v}, \vec{\theta}) ~=~&
2\sum_{E\in\mathcal{E}_h} \int_E \strain(\vec{v}):\strain(\vec{\theta}) - 2\sum_{e\in\Gammah\cup\partial\Omega^\mathrm{D}} \int_e \avg{\strain(\vec{v})\cdot\vec{n}_e} \cdot \jump{\vec{\theta}}\\
&+2\epsilon \sum_{e\in\Gammah\cup\partial\Omega^\mathrm{D}} \int_e \avg{\strain(\vec{\theta})\cdot\vec{n}_e} \cdot \jump{\vec{v}} + \frac{\sigma}{h} \sum_{e\in\Gammah\partial\Omega^\mathrm{D}}\int_e \jump{\vec{v}}\cdot\jump{\vec{w}}~,
\end{split}\\
b_\strain(\vec{\theta}) ~=~&
2\epsilon \sum_{e\in\partial\Omega^\mathrm{D}} \int_e \big(\strain(\vec{\theta})\cdot\vec{n}_e\big)\cdot \vec{v}_\mathrm{D} + \frac{\sigma}{h} \sum_{e\in\partial\Omega^\mathrm{D}}\int_e \vec{v}_\mathrm{D}\cdot\vec{\theta}~.
\end{align*}
In addition, the pressure update postprocessing step \cref{eq:NS:time_dis_post_step_p} should be changed to 
\begin{align*}
p^n ~&=~ p^{n-1} + \phi^n - \frac{\mathrm{X}}{\Rey}\,\div{\vec{v}^n}~.
\end{align*}
The corresponding \cref{eq:NS:space_dis_post_step_p}, in space discretization part, should also be changed to 
\begin{align*}
(p^n_h,\chi_h) ~&=~ (p^{n-1}_h,\chi_h) + (\phi^n_h,\chi_h) - \frac{\mathrm{X}}{\Rey}\,(\div{\vec{v}^n_h},\chi_h)~.
\end{align*}
Note, the above equations comprise all of the changes that are necessary when modifying the numerical scheme.

\section*{References}
\bibliography{bibliography} 
\bibliographystyle{elsarticle-num}

\end{document}